\documentclass[12pt]{article}
\usepackage[utf8]{inputenc}
\usepackage{physics}
\usepackage{parskip}
\usepackage{lscape}
\usepackage{hyperref}
\usepackage{tcolorbox}
\usepackage{caption}
\usepackage{amsmath}
\usepackage{amsthm}
\usepackage{amssymb}
\usepackage{relsize}
\usepackage{mathrsfs}
\usepackage{braket}
\usepackage{mathtools}
\usepackage[symbol]{footmisc}
\usepackage[letterpaper, portrait, margin=1in]{geometry}
\usepackage{lmodern}
\usepackage[T1]{fontenc}
\usepackage{graphicx}
\graphicspath{{./}}
\usepackage{csquotes}
\usepackage{environ}
\usepackage{trimspaces}
\usepackage{tikz}

\theoremstyle{plain}
\newtheorem{theorem}{Theorem}[section]
\newtheorem{corollary}[theorem]{Corollary}
\newtheorem{proposition}[theorem]{Proposition}

\theoremstyle{definition}
\newtheorem{definition}[theorem]{Definition}

\numberwithin{equation}{section}

\title{Quasi-rectifiable Lie algebras \\ for partial differential equations}
\author{A.M. Grundland$^{*,**}$ and J. de Lucas\footnote{Corresponding author: Javier de Lucas. Email: j.de-lucas-ara@uw.edu.pl. }$^{*,***}$\footnote{{\bf Declarations of interest: none}} \\$^*$Centre de Recherches Mathématiques, Université de Montréal,\\Succ. Centre-Ville, CP 6128, Montréal (QC) H3C 3J7, Canada.\\$^{**}$D\'epartement de Math\'ematiques et d'Informatique, Université du Québec,\\ CP 5000, Trois-Rivières (QC), G9A 5H7, Canada.\\ $^{***}$Department of Mathematical Methods in Physics, University of Warsaw,\\ ul. Pasteura 5, 02-093, Warszawa, Poland.}
\date{}

\begin{document}
\maketitle
\begin{abstract}
    We introduce families of quasi-rectifiable vector fields and study their geometric and algebraic aspects. Then, we analyse their applications to systems of partial differential equations. Our results explain, in a simpler manner, previous findings about hydrodynamic-type equations. Facts concerning families of quasi-rectifiable vector fields, their relation to Hamiltonian systems, and practical procedures for studying such families are developed. We introduce and analyse quasi-rectifiable Lie algebras, which are motivated by geometric and practical reasons. We classify different types of quasi-rectifiable Lie algebras, e.g. indecomposable ones up to dimension five. New methods for solving systems of hydrodynamic-type equations are established to illustrate our results. In particular, we study hydrodynamic-type systems admitting $k$-wave solutions through quasi-rectifiable Lie algebras of vector fields. We develop techniques for obtaining the submanifolds related to quasi-rectifiable Lie algebras of vector fields and systems of partial differential equations admitting a nonlinear superposition rule: the PDE Lie systems.
\end{abstract}
{\bf Keywords:} systems of partial differential equations, Lie algebra, quasi-rectifiable Lie algebras, classification of Lie algebras, PDE Lie system.

{\bf MSC 2020:} 35Q53 (primary); 35A30, 35Q58, 53A05 (secondary).

\newpage
\tableofcontents
\section{Introduction}
The study of hydrodynamic-type equations \cite{Pe72,RY83,Za79,Za80} via the method of characteristics using Riemann invariants includes a description of certain integrability conditions under which the method is applicable \cite{Je76,Ma84,Mi58,Pe85,Po08,Ri58,Za79,Za80}. In particular, the method requires the existence of a family of vector fields $\{X_1,\ldots,X_r\}$ on the manifold $N$ describing the dependent variables of the hydrodynamic-type equations such that each Lie bracket $[X_i,X_j]$ is a linear combination, with certain arbitrary functions on $N$, of $X_i$ and $X_j$. We will call these families of vector fields {\it quasi-rectifiable} due to their mathematical properties, or {\it elastic} because of their relationship with the method of characteristics and hydrodynamic-type equations associated with Riemann invariants. In the literature, the method of characteristics involving Riemann invariants assumes that $X_1,\ldots,X_r$ can be rescaled by multiplying such vector fields with functions $f_1,\ldots,f_r$ so that $[f_iX_i,f_jX_j]=0$ for $1\leq i<j\leq r$ (see \cite{GV91} for further details). 

The interest in the modified Fr\"obenius theorem by rescaling is due to its ability to demonstrate the viability of certain practical simplifications within the method of characteristics involving Riemann invariants \cite{Gr74II,GL23,GV91,Pe74}. This is used to obtain, in a simple manner, a system of partial differential equations (PDEs) whose solutions describe parametrisations of $k$-wave solutions of hydrodynamic-type equations \cite{GT96,GV91,GZ83a,GZ83b}. More specifically, the work \cite[p. 239]{GV91} claims without proof that a simple calculation could rescale a quasi-rectifiable family of vector fields $X_1,\ldots,X_r$ in a manner that would lead to the commutation of the rescaled vector fields $f_1X_1,\ldots,f_rX_r$ between themselves. However, if $r>2$, the functions $f_1,\ldots,f_r$ must satisfy a complicated system of $r(r-1)/2$ partial differential equations whose solution must be proven to exist. The existence of a rescaling is guaranteed by the modified Fr\"obenius theorem by rescaling \cite{GL23}. This rescaling significantly simplifies the method of characteristics involving Riemann invariants. It should be noted that the Fr\"obenius theorem cannot be used to rescale $X_1,\ldots,X_r$ as needed in that method because it merely guarantees that the integrable distribution $\mathcal{D}$ spanned by $X_1,\ldots,X_r$ is generated by some commuting vector fields $Y_1,\ldots,Y_r$ taking values in $\mathcal{D}$. However, the vector fields $Y_1,\ldots,Y_r$ obtained by the Fr\"obenius theorem are derived from certain linear combinations of $X_1,\ldots,X_r$. In contrast, within the generalised method of characteristics, the vector fields $X_1,\ldots,X_r$ can only be rescaled to commute, rendering the Fr\"obenius theorem inadequate to ensure their rescaling. Note that the above differences with respect to the Fr\"obenius theorem and its fields of application explain why the modified Fr\"obenius theorem by rescaling \cite{GL23} has not previously been studied in depth. 

Despite the significance of the modified Fr\"obenius theorem by rescaling, its statement only ensures the existence of $f_1,\ldots,f_r$. In other words, it does not give a method for their explicit determination. Moreover, one may wonder about the theoretical properties of quasi-rectifiable families of vector fields. In particular, one may try to look for a coordinate system putting the vector fields in a `canonical', so called {\it quasi-rectifiable form}.  Analysing this problem is the first aim of this paper. 

As a first result, we formalise the notion of a quasi-rectifiable family of vector fields. This is used to put in a rigorous way what has already been used in the literature in an intuitive manner (cf. \cite[pg. 349]{Pe74}). A precise definition allows us to provide a new, simpler approach to the modified Fr\"obenius theorem by rescaling. Consequently, we introduce a new concept of {\it quasi-rectifiable Lie algebras of vector fields}, which is studied in this paper and whose applications in hydrodynamic-type equations are investigated. Quasi-rectifiable Lie algebras of vector fields are associated with elastic superposition of waves described by hydrodynamic-type systems (\cite{Pe85}).

Next, we develop new techniques for obtaining the associated functions $f_1,\ldots,f_r$. This has immediate applications in the theory of the generalised method of characteristics involving Riemann invariants, as it is done so as to obtain an appropriate parametrisation of solutions \cite{Gr23,GL23,GV91}. Some of our new methods involve the generalisation of a classical method for solving systems of PDEs  \cite{Sn06}. This involves a new application of the so-called evolution vector fields in contact geometry, which appeared recently in \cite{SLLM20} as a method for the study of thermodynamic systems, and has been receiving attention \cite{ELLM21,LLM21}. Moreover, our new techniques have applications,  such as in the theory of Lie systems in order to obtain superposition rules for certain classes of Lie systems. This avoids the necessity of putting Lie systems in canonical form as in \cite{Wi83} or of using geometric structures as in the Poisson coalgebra method in \cite{LS20}. Our method is  simpler than the standard methods based on the integration of families of vector fields and the solving of algebraic equations \cite{Wi83} or the solving of systems of PDEs \cite{CGM07}. Moreover, as a consequence, Proposition \ref{Prop:NewMethod} provides new approaches for obtaining solutions of systems of PDEs of the form (\ref{eq:DifEq2}).  Some of our techniques have applications in the solving of general systems of first-order PDEs. Moreover, they are also useful for the determination of properties related to quasi-rectifiable Lie algebras of vector fields. For example, Theorem \ref{Th:FroRecBett} and Corollary \ref{Cor:Double} provide the functions rescaling a quasi-rectifiable family of vector fields into commuting vector fields as well as the coordinates putting such vector fields into quasi-rectifiable form.

We study the relation between quasi-rectifiable families of vector fields and the integration of systems of ordinary differential equations. In particular, their appearance in the study of Lie symmetries and Sundman transformations for systems of ordinary differential equations  is studied. Additionally, it is determined when a quasi-rectifiable family of vector fields consists of Hamiltonian vector fields and when it can be put into a quasi-rectifiable form that consists of Hamiltonian vector fields as well. It is proved that Hamiltonian families of quasi-rectifiable vector fields admit families of Hamiltonian functions of a particular type, which reassembles the commutation relations appearing in the theory of Poisson algebra deformations of the so-called Lie--Hamilton systems \cite{BCFHL17}. This gives rise to the definition of quasi-rectifiable families of Hamiltonian functions.  In particular, integrable Hamiltonian systems (in a symplectic sense) give rise to quasi-rectifiable families of Hamiltonian functions of a very particular type.

Additionally, we define {\it quasi-rectifiable Lie algebras}. These are Lie algebras admitting a basis $\{e_1,\ldots,e_r\}$ such that each commutator $[e_i,e_j]$ is spanned by $e_i$ and $e_j$. We study the properties of such Lie algebras and classify quasi-rectifiable Lie algebras of dimension up to five, considering the case of indecomposable Lie algebras of dimension four and five. Moreover, other cases of higher-dimensional quasi-rectifiable Lie algebras are studied. 

Finally, some applications of our techniques appearing in the theory of hydrodynamic-type equations are analysed. In particular, our study is first concerned with hydrodynamic equations on a $(1+1)$-dimensional manifold. This case is related to a quasi-rectifiable Lie algebra of vector fields and our methods are applied to put a basis of such a Lie algebra into a quasi-rectifiable form. Then, $k$-wave solutions for the hydrodynamic equations of a barotropic fluid in $(1+1)$-dimensions \cite{Gr23} are found. In particular, this illustrates the existence of certain quasi-rectifiable Lie algebras of vector fields and the usefulness of our methods for studying practical problems. Finally, we provide a method that allows one to construct systems of PDEs admitting $k$-wave solutions. As an example, our classification of quasi-rectifiable Lie algebras is used to obtain one system of PDEs admitting a three-wave solution related to a quasi-rectifiable Lie algebra of vector fields isomorphic to $\mathfrak{r}_{3,-1}$ given in Table \ref{fig:Lie-algebra-classification}.  Moreover, instead of putting a basis of such a Lie algebra into quasi-rectifiable form, as classically done in the generalised method of characteristics, we provide a so-called Lie system of PDEs \cite{CGM07,CL11} in order to obtain a parametrisation of the $k$-wave solutions. The properties of this system of PDEs is related to the structure of a quasi-rectifiable Lie algebra of vector fields, which also justifies the use of our classification in Section \ref{Sec::ClaRecLieAlg}.

This paper is structured as follows. Section \ref{Sec:Elastic} is devoted to the study of quasi-rectifiable Lie algebras of vector fields and several methods related to them. Moreover, theoretical applications of our results to contact geometry, nonlinear superposition rules, and hydrodynamic-type systems are described. Section \ref{Sec:Direct} analyses direct methods for putting families of vector fields in a quasi-rectifiable form. Section \ref{Sec::Int} analyses the relation between quasi-rectifiable families of vector fields and Lie symmetries of ordinary differential equations and integrable Hamiltonian systems relative to symplectic manifolds. Section \ref{Sec::ClaRecLieAlg} analyses abstract quasi-rectifiable Lie algebras and classifies them for several types of Lie algebras. Finally, some specific applications of our results to hydrodynamic-type equations are presented in Section \ref{Se:Hydro}. Our classification of indecomposable four- and five-dimensional quasi-rectifiable Lie algebras is summarised in tables in the Appendix.

\section{Quasi-rectifiable families of vector fields}\label{Sec:Elastic}

In the study of the solving of hydrodynamic-type equations via Riemann invariants, an interesting concept appears: families of vector fields satisfying a certain class of commutation relations \cite{GL23,Ri58}. Their definition and the study of their existence and main properties is the aim of this section. To stress the main points of our presentation, we will assume all structures to be globally defined and smooth. % Moreover, a related notion, the elastic Lie algebras of vector fields, appears and is analysed. 
Note that, in what follows, we do not use the Einstein notation over repeated indices. For simplicity, all structures are assumed to be globally defined and smooth unless otherwise stated. From now on, given a list $a_1,\ldots,a_k$ of $k$ elements, $a_1,\ldots,\widehat{a}_i,\ldots, a_k$ stands for the list of $(k-1)$-elements obtained by skipping the term $a_i$ in the previous one. Moreover, $N$ stands for an $n$-dimensional manifold.

Let us start by giving a new geometric characterisation of a relevant class of families of vector fields appearing in the Riemann invariants method \cite{GV91}.

\begin{theorem}\label{Th::RecBasis}Let $X_1,\ldots ,X_r$ be a family of vector fields on  $N$ such that $X_1\wedge\ldots\wedge X_r$ does not vanish on $N$. There exists a  coordinate system $\{x^1,\ldots,x^n\}$  on $N$ such that the integral curves of each $X_i$ are given by $x^1=k_1,\ldots,{x}^{i-1}=k_{i-1},{x}^{i+1}=k_{i+1},\ldots,x^n=k_n$  for some constants $k_1,\ldots,\hat{k}_{i},\ldots, k_n\in \mathbb{R}$, if and only if
\begin{equation}\label{eq:CommRel}
[X_i,X_j]=f_{ij}^iX_i+f_{ij}^jX_j,\qquad 1\leq i<j\leq r,
\end{equation}
for a family of $r(r-1)$ functions $f_{ij}^i,f_{ij}^j\in C^\infty(N)$ with $1\leq i<j\leq r$.
\end{theorem}
\begin{proof} If the coordinate system $\{x^1,\ldots,x^n\}$ exists, $X_ix^j=0$ for $i\neq j$, $i=1,\ldots,r$ and $j=1,\ldots,n$. Hence,
\begin{equation}\label{eq:AlmRecVec}
    X_i=
g^i(x^1,\ldots,x^n)\frac{\partial}{\partial x^i},\qquad i=1,\ldots,r,
\end{equation}
for some functions $g^1,\ldots,g^r:N\rightarrow \mathbb{R}$. Then, the relations (\ref{eq:CommRel}) follow. 

Conversely, if the relations (\ref{eq:CommRel}) hold, then the distribution spanned by $X_1,\ldots,X_r$ has rank $r$ by assumption and it admits $n-r$ first integrals $x^{r+1},\ldots,x^n$ that are common for $X_1,\ldots,X_r$ and  functionally independent, i.e. $dx^{r+1}\wedge \ldots\wedge dx^n$ does not vanish at any point of $N$. Moreover, each distribution 
$$
\mathcal{D}^{(i)}_x=\langle X_1(x),\ldots,\hat{X}_{i}(x),\ldots,X_r(x)\rangle,\qquad x\in N,\qquad i=1,\ldots,r,
$$  
is integrable and has rank $r-1$. Since $N$ is an $n$-dimensional manifold, the vector fields $X_1,\ldots,\hat{X}_{i},\ldots,X_r$  admit a common non-constant first integral $x^i$, i.e. $X_jx^i=0$ for $j=1,\ldots,\hat{i},\ldots,r$, such that $\Upsilon^i=dx^i\wedge dx^{r+1}\wedge\ldots\wedge dx^n$ is not vanishing. Note that $\mathcal{D}^{(i)}$ is the distribution spanned by the vector fields $X$ on $N$ taking values in the kernel of $\Upsilon^i$, namely $\iota_X\Upsilon^i=0$, where $\iota_X\Upsilon^i$ stands for the contraction of the vector field $X$ with the differential one-form $\Upsilon^i$. Moreover, $\iota_{X_i}\Upsilon^i=(X_ix^i)dx^{r+1}\wedge \ldots \wedge dx^n\neq 0$ and $X_ix^i\neq 0$. Since the contractions of the vector fields $X_1,\ldots,X_r$ with $dx^1\wedge\ldots\wedge dx^r$ satisfy
$$
\iota_{X_r}\ldots\iota_{\hat{X}_{i}}\ldots\iota_{X_1}dx^1\wedge\ldots\wedge dx^r=\left[(-1)^{r-i}\prod_{1\leq j\neq i\leq r}(X_jx^j)\right]dx^i\neq 0,\qquad i=1,\ldots,r,
$$
it follows that $dx^1\wedge \ldots\wedge dx^n$ is a volume form and the coordinates $\{x^1,\ldots,x^n\}$ form a local coordinate system on $N$.  On the coordinate system $\{x^1,\ldots,x^n\}$, the vector fields $X_1,\ldots,X_r$ take the form (\ref{eq:AlmRecVec}) and the converse part of our theorem follows.
\end{proof}

Theorem \ref{Th::RecBasis} and the results in this section motivate the following definition.
\begin{definition} A family of vector fields $X_1,\ldots,X_r$ on $N$ is said to be {\it quasi-rectifiable} if there exists a coordinate system $\{x^1,\ldots,x^n\}$ on $N$ such that
\begin{equation}\label{eq:ConElaCoor}
X_i=g^i(x^1,\ldots,x^n)\frac{\partial}{\partial x^i},\qquad i=1,\ldots,r,\qquad \prod_{i=1}^rg^i(x^1,\ldots,x^n)\neq 0,
\end{equation}
for some functions $g^1,\ldots,g^r:N\rightarrow \mathbb{R}$. 
Otherwise, the family $X_1,\ldots,X_r$ is called {\it non quasi-rectifiable}. The coordinate expression (\ref{eq:ConElaCoor}) is called a {\it quasi-rectifiable} form for $X_1,\ldots,X_r$.
	
\end{definition}

For hydrodynamic-type systems, the terms {\it elastic} and {\it inelastic} are used instead of quasi-rectifiable and non quasi-rectifiable, respectively, due to the presence of nonlinear superpositions of Riemann waves \cite{Pe85}. Indeed, the terms elastic and inelastic were used, without a precise definition, in the literature 
(see for instance \cite[pg. 349]{Pe74}). Theorem \ref{Th::RecBasis} is ``optimal'' in the sense that  if the commutator of  two vector fields of the family $X_1,\ldots,X_r$ is not spanned by two such vector fields or $X_1\wedge \ldots\wedge X_r$ vanishes at a point, then the existence of the coordinates $\{x^1,\ldots,x^n\}$ is not ensured. Let us illustrate this fact with several examples.

Consider the Heisenberg matrix Lie group 
$$
\mathbb{H}_3=\left\{\left(\begin{array}{ccc}1&x&y\\0&1&z\\0&0&1\end{array}\right):x,y,z\in \mathbb{R}\right\}
$$
with coordinates $\{x,y,z\}$ and the vector fields on $\mathbb{H}_3$ given by
$$
X_1=\frac{\partial}{\partial x},\qquad X_2=\frac{\partial}{\partial y}+x\frac{\partial}{\partial z},\qquad X_3=\frac{\partial}{\partial z}.
$$
Then, $X_1\wedge X_2\wedge X_3$ does not vanish on $\mathbb{H}_3$ and 
$$
[X_1,X_2]=X_3,\qquad [X_1,X_3]=0,\qquad [X_2,X_3]=0.
$$
Let us prove that $X_1,X_2,X_3$ are not quasi-rectifiable by contradiction. Assume that there exists a coordinate system $\{x^1,x^2,x^3\}$ as in Theorem \ref{Th::RecBasis}. Then, $x^3$ must be a common first integral of $X_1,X_2$. Since $[X_1,X_2]=X_3$, it follows that
$$
X_1x^3=X_2x^3=0\Rightarrow 0=X_1X_2x^3-X_2X_1x^3=X_3x^3=0.
$$
Hence, $x^3$ is a first integral of $X_1,X_2,X_3$, and it becomes a constant because $X_1,X_2,X_3$ span $T_{(x,y,z)}\mathbb{H}^3$ for every $(x,y,z)\in \mathbb{R}^3$.  This is a contradiction and the vector fields $X_1,X_2,X_3$ are non quasi-rectifiable.

Let us study a second example. Consider the linear coordinates $\{x,y,z\}$ on $\mathbb{R}^3$ and the family of vector fields on $\mathbb{R}^3$ given by
$$
X_1=y\frac{\partial}{\partial x}-x\frac{\partial}{\partial y},\qquad X_2=z\frac{\partial}{\partial y}-y\frac{\partial}{\partial z},\qquad X_3=x\frac{\partial}{\partial z}-z\frac{\partial}{\partial x},
$$
namely the infinitesimal generators of the clock-wise rotations in $\mathbb{R}^3$ around the $Z$, $X$, and $Y$ axes, respectively. It is immediate that $X_1\wedge X_2\wedge X_3=0$. Indeed, $X_1,X_2,X_3$ admit a common first integral $x^2+y^2+z^2$. Hence, $X_1,X_2,X_3$ are not quasi-rectifiable since the quasi-rectifiable form (\ref{eq:ConElaCoor}) implies that the elements of a family of quasi-rectifiable vector fields are linearly independent at each point of the manifold.

There is another way to understand Theorem \ref{Th::RecBasis}. Consider that $X_1\wedge\ldots\wedge X_r$  does not vanish at any point. The existence of a coordinate system $\{x^1,\ldots,x^n\}$ satisfying the given conditions (\ref{eq:ConElaCoor}) implies that there exist non-vanishing functions $g^1,\ldots,g^r$ ensuring that $X_1/g_1,\ldots,X_r/g_r$ commute between themselves. Conversely, if the vector fields $X_1/g_1,\ldots,X_r/g_r$ commute between themselves, then there exist coordinates $\{x^1,\ldots,x^n\}$ such that the previous vector fields can be simultaneously rectified
$$
X_i/g_i=\frac{\partial}{\partial x^i},\qquad i=1,\ldots,r,
$$
which shows that $\{x^1,\ldots,x^n\}$ satisfies the required conditions. Hence, this proves  the following theorem, which was demonstrated in \cite{GL23} in another manner.
\begin{theorem}\label{Th::InLiterature}
Given a family of vector fields $X_1,\ldots,X_r$ defined on   $N$ such that $X_1\wedge\ldots\wedge X_r$ does not vanish, there exist non-vanishing functions $h_1,\ldots,h_r$ such that the $h_1X_1,\ldots,h_rX_r$ commute between themselves, i.e.
$$
[h_iX_i,h_jX_j]=0,\qquad 1\leq i<j\leq r,
$$
if and only if the conditions (\ref{eq:CommRel}) hold for a family of $r(r-1)$ functions $f^i_{ij},f^j_{ij}$ on $N$ with  $1\leq i<j\leq r$.
\end{theorem}

Theorem \ref{Th::InLiterature} appears in the theory of hydrodynamic-type equations, which initially motivated the present work. Some applications of the results of this section will be discussed in Section \ref{Se:Hydro}. It is remarkable that the proof of Theorem \ref{Th::InLiterature} follows by induction. The result is immediate for $r=2$. Next, one assumes that the result is satisfied for $r-1$ vector fields and   that $X_1,\ldots,X_{r-1}$ have already been rescaled to commute among themselves. Note that if $X_1,\ldots,X_r$ satisfy the conditions (\ref{eq:CommRel}), the rescaling of $X_1,\ldots,X_{r-1}$ to make them commute between themselves gives rise to new vector fields $Y_1,\ldots,Y_{r-1}$ that commute between themselves and span the same distribution as $X_1,\ldots,X_{r-1}$. Then, $Y_1,\ldots,Y_{r-1},X_r$ satisfy (\ref{eq:CommRel}) relative to new functions $f_{ij}^{'i},f_{ij}^{'j}$, with $1\leq i< j\leq r$. Next, one multiplies $X_r$ by a function $h_r$ so that $h_rX_r$ commutes with $f_1Y_1,
\ldots,f_{r-1}Y_{r-1}$, where the functions $f_1,\ldots, f_{r-1}$ are chosen so that they are first integrals for the vector fields taking values in the distribution spanned by $Y_1,\ldots,Y_{r-1}$. Note that one multiplies $X_r$ by a non-vanishing function so that the flow of the vector field $h_rX_r$ leaves the distribution spanned by $X_1,\ldots,X_{r-1}$ invariant. At the end, one finds that the original vector fields $X_1,\ldots,X_r$ must be multiplied by non-vanishing functions $h_1,\ldots,h_r$ so as to make them commute.

It is interesting to remark that one can define a type of Lie algebra of vector fields admitting a basis that can be written in quasi-rectifiable form. The practical relevance of these Lie algebras of vector fields will be justified in Section \ref{Se:Hydro}, and it involves, for instance, the study of linear systems of PDEs and the Riemann invariants method for hydrodynamic-type equations. 

\begin{definition} A Lie algebra $V$ of vector fields on a manifold $N$ is \textit{quasi-rectifiable} if it admits a basis $
\{X_1,\ldots,X_r\}$ such that $X_1\wedge\ldots\wedge X_r$ does not vanish on $N$ and the Lie bracket of any pair $X_i,X_j$ is a linear combination of $X_i$ and $X_j$ with constant coefficients, i.e. $[X_i,X_j]\wedge X_i\wedge X_j=0$ for $1\leq i<j\leq r$.
\end{definition}

Since the vector fields $X_1,\ldots,X_r$ giving a basis of the Lie algebra of vector fields $V$ in the above definition are assumed to be such that $X_1\wedge \ldots\wedge X_r$ does not vanish at any point, it follows that if $[X_i,X_j]=f_{ij}^iX_i+f_{ij}^jX_j$ for certain functions $f_{ij}^i,f_{ij}^j\in C^\infty(N)$, the decomposition $[X_i,X_j]=\sum_{k=1}^r f_{ij}^kX_k$ is unique and the functions $f_{ij}^k$ are constants because $X_1,\ldots,X_r$ span a Lie algebra of vector fields. Moreover, it may happen that  a basis $\{X_1,\ldots,X_r\}$ of $V$ is quasi-rectifiable and another basis of $V$ is not. It is also worth noting that, in view of Theorem \ref{Th::RecBasis}, a Lie algebra of vector fields is quasi-rectifiable if and only if it admits a basis that can be written in the form (\ref{eq:ConElaCoor}).

Consider the matrix Lie group $SL_2(\mathbb{R})$ of $2\times 2$ real matrices with determinant one
\begin{equation}\label{Eq:SL2}
SL_2(\mathbb{R})=\left\{\left(\begin{array}{cc}a&b\\c&d\end{array}\right):ad-bc=1,a,b,c,d\in \mathbb{R}\right\}.
\end{equation}
Here, $\{a,b,c\}$ forms a local coordinate system of $SL_2(\mathbb{R})$ close to its neutral element. Thus, a basis of the space of left-invariant vector fields on $SL_2(\mathbb{R})$ may be chosen to be		$$
		X^L_1 = a\frac{\partial}{\partial a} - b \frac{\partial}{\partial b} + c \frac{\partial}{\partial c}, \quad X^L_2 = a \frac{\partial}{\partial b}, \quad X^L_3 = b \frac{\partial}{\partial a} + \left(\frac{1+bc}{a}\right) \frac{\partial}{\partial c}.
		$$
		These vector fields on $SL_2(\mathbb{R})$ satisfy the commutation relations for a basis of the matrix Lie algebra of traceless $2\times 2$ matrices and span the Lie algebra $\mathfrak{sl}_2(\mathbb{R})$ of $SL_2(\mathbb{R})$, namely
		\begin{equation}\label{ConRel}
		[X^L_1,X^L_2]=2X^L_2,\qquad [X^L_1,X^L_3]=-2X^L_3,\qquad [X^L_2,X^L_3]=X^L_1.
		\end{equation}
  Note that $X^L_1\wedge X^L_2\wedge X^L_3$ does not vanish at any point in $SL_2(\mathbb{R})$.	Due to (\ref{ConRel}), the basis $\{X^L_1,X^L_2,X^L_3\}$ is not in quasi-rectifiable form. However, let us choose a new basis of $\mathfrak{sl}_2(\mathbb{R})$ given by
  \begin{equation}\label{Eq:basis_sl2}
Y_1^L= \frac{ X_1^L}{2},\qquad Y_2^L=\frac{1}{\sqrt{2}}X_2^L+\frac{1}{4}X^L_1,\qquad Y_3^L=\frac{1}{\sqrt{2}}X_3^L-\frac{1}{2}X_1^L.
  \end{equation}
  Indeed,
  \begin{equation}\label{Eq:ComRel}
  [Y_1^L,Y_2^L]=Y_2^L-\frac 1 2Y_1^L,\qquad [Y_1^L,Y_3^L]=-Y_3^L-Y_1^L,\qquad [Y_2^L,Y_3^L]=Y_2^L-\frac 1 2Y^L_3.
  \end{equation}Then, $Y_1^L,Y_2^L,Y_3^L$ satisfy that $Y_1^L\wedge Y_2^L\wedge Y_3^L$ does not vanish and the conditions (\ref{eq:CommRel}) hold. Hence, $\langle Y_1^L,Y_2^L,Y_3^L\rangle$ is a quasi-rectifiable Lie algebra of vector fields and the basis (\ref{Eq:basis_sl2}) is in quasi-rectifiable form. Indeed, $Y_1^L,Y_2^L,Y_3^L$ become a quasi-rectifiable family of vector fields. One may wonder how we obtained $Y_1^L,Y^L_2,Y^L_3$. The answer will be  given in Theorem \ref{Th::RectCri}. In particular, the basis will be derived by obtaining three particular solutions of the algebraic equation (\ref{Eq:ConRec}) for $\mathfrak{sl}_2(\mathbb{R})$, which are straightforward to obtain.
\section{Methods for constructing quasi-rectifiable families of vector fields}\label{Sec:Direct}

In the previous section, we developed a formalism to study families of quasi-rectifiable vector fields and quasi-rectifiable Lie algebras of vector fields. Nevertheless, the given approach was mainly theoretical and the application of these notions and results to practical cases requires us to put a quasi-rectifiable family of vector fields into a quasi-rectifiable form. The aim of this section is to develop practical methods to accomplish this result and to solve other related problems.  

Let us illustrate how to apply Theorem \ref{Th::RecBasis} to the particular case of the basis (\ref{Eq:basis_sl2}) of left-invariant vector fields on $SL_2(\mathbb{R})$. In the coordinates $\{a,b,c\}$ of $SL_2(\mathbb{R})$ appearing in \eqref{Eq:SL2}, the vector fields $Y^L_1,Y_2^L,Y^L_3$ are
\begin{equation}\label{Eq:CoorSL2}
\begin{gathered}
Y^L_1=\frac{a}{2}\frac{\partial}{\partial a}-\frac{b}{2}\frac{\partial}{\partial b}+\frac{c}{2}\frac{\partial}{\partial c},\qquad Y^L_2=\frac{a}{4}\frac{\partial}{\partial a}+\left(\frac{a}{\sqrt{2}}-\frac{b}{4}\right)\frac{\partial}{\partial b}+\frac{c}{4}\frac{\partial}{\partial c},\\
Y^L_3=\left(\frac{b}{\sqrt{2}}-\frac{a}{2}\right)\frac{\partial}{\partial a}+\frac{b}{2}\frac{\partial}{\partial b}+\left[\frac{1}{\sqrt{2}}\left(\frac{1+bc}{a}\right)-\frac{c}{2}\right]\frac{\partial}{\partial c}.
\end{gathered}
\end{equation}
From (\ref{Eq:ComRel}) and using the method of characteristics \cite{Sn06}, one finds a common first integral $x^3$ for the vector fields $Y^L_1,Y^L_2$, a common first integral $x^2$ for the vector fields $Y^L_1,Y^L_3$, and a common first integral, $x^1$, for the vector fields $Y^L_2,Y^L_3$. Such first integrals are, for instance,
\begin{equation*}%\label{Eq:CoordSys}
x^1=\frac{c}{\sqrt{2}a}-\frac{1}{2a^2-\sqrt{2} ab},\qquad x^2=\frac{1+bc}{ab},\qquad x^3=\frac{c}{a}.
\end{equation*}
Using the coordinates $\{x^1,x^2,x^3\}$, the vector fields $Y^L_1,Y^L_2,Y^L_3$ can be brought into the form
$$
\begin{gathered}
Y^L_1=\frac{2}{(-2a+\sqrt{2}b)^2}\frac{\partial}{\partial x^1},\qquad Y^L_2=-\frac{1}{\sqrt{2}b^2}\frac{\partial}{\partial x^2},\qquad 
Y^L_3=\frac{1}{\sqrt{2}a^2}\frac{\partial}{\partial x^3},
\end{gathered}
$$
where the coefficient functions of the previous vector fields have been expressed in terms of the coordinate functions $a,b,c$ in order to simplify the obtained expressions.
Hence, the multiplication of $Y^L_1,Y^L_2,Y^L_3$ by the functions 
\begin{equation}\label{Eq:RecFroFun}
h_1=\frac{(-2a+\sqrt{2} b)^2}{2},\qquad h_2=-\sqrt{2}b^2,\qquad h_3=\sqrt{2}a^2 
\end{equation}
give,respectively, new vector fields proportional to $Y^L_1,Y^L_2,Y^L_3$ which commute between themselves. 

The previous method for the determination of the coordinates $\{x^1,\ldots,x^n\}$ requires the calculation of first integrals for families of vector fields using the method of characteristics. In fact, this can be seen in the proof of Theorem \ref{Th::RecBasis}. In order to obtain some common first integrals of $X_1,\ldots,\hat{X}_{i},\ldots, X_r$, one uses a maximal set of functionally independent first integrals for $X_1$ obtained by the method of characteristics. Then, one writes the remaining vector fields in terms of a coordinate system consisting of these first integrals and some additional variables. Assuming that the action of $X_2$ on the coordinates that correspond to first integrals of $X_1$ vanishes, the procedure can be applied successively. 

It is worth stressing that the derivation of common first integrals for families of vector fields also appears in the study of nonlinear superposition rules for systems of first-order ordinary differential equations (ODEs) \cite{CL11} and in the determination of Darboux coordinates for geometric structures \cite{GLRR23}. Let us now give a generalisation of a method for obtaining such constants of motion. 

Let us denote the first-jet manifold, $J^1(N,N\times \mathbb{R})$, of sections relative to the projection $\pi:(x,t)\in N\times \mathbb{R}\mapsto x\in N$ simply as $J^1\pi$. Then, $J^1\pi$ is endowed with a canonical contact structure, namely a maximally non-integrable distribution of co-rank one, given by the Cartan distribution of $J^1\pi$ (see \cite{Ar90}).  In coordinates adapted to $J^1\pi$, say $\{x^1,\ldots,x^n,z,p_1,\ldots,p_n\}$, the Cartan distribution is given locally  by the vector fields taking values in the kernel of the one-form $\eta=dz-\sum_{i=1}^np_idx^i$, i.e.
$$
\mathcal{C
}=\left\langle \frac{\partial}{\partial x^1}+p_1\frac{\partial}{\partial z},\ldots,\frac{\partial}{\partial x^n}+p_n\frac{\partial}{\partial z},\frac{\partial}{\partial p_1},\ldots,\frac{\partial}{\partial p_n}\right\rangle.
$$
In particular, we are interested in finding contact geometry methods allowing us to obtain a non-constant solution of the PDE system
\begin{equation}\label{eq:DifEq}
X_{(i)}f=g_i,\qquad i=1,\ldots,r,
\end{equation}
for a family of vector fields $X_{(1)},\ldots,X_{(r)}$ on $N$ spanning a distribution $\mathcal{D}$ of rank $r$ and some functions $g_1,\ldots,g_r$ depending on $N$ and possibly on $f$. In particular, if $g_1,\ldots,g_r=0$,  it is known that a non-constant $f$ exists if and only if the smallest integrable distribution $\mathcal{D}'$ containing $\mathcal{D}$ 
has rank $r'<\dim N$. 

 In the adapted coordinates $\{x^1,\ldots,x^n,z,p_1,\ldots,p_n\}$ of $J^1\pi$, the system of PDEs (\ref{eq:DifEq}) with $g_1=,\ldots,g_r=0$, can be rewritten as follows
$$
f_j\left(x^1,\ldots,x^n,f(x^1,\ldots,x^n),\frac{\partial f}{\partial x^1}(x^1,\ldots,x^n),\ldots,\frac{\partial f}{\partial x^n}(x^1,\ldots,x^n)\right)=0,\qquad j=1,\ldots,r
$$
for\begin{equation*}%\label{Def:fun}
f_j(x^1,\ldots,x^n,z,p_1,\ldots,p_n)=\sum_{i=1}^nX_{(j)}^i(x^1,\ldots,x^n)p_i,\qquad j=1,\ldots,r,
\end{equation*}
where 
$$
X_{(j)}=\sum_{i=1}^nX_{(j)}^i(x^1,\ldots,x^n)\frac{\partial}{\partial x^i},\qquad j=1,\ldots,r.
$$
The fact that $X_1\wedge\ldots\wedge X_r$ does not vanish in a neighbourhood of $x=(x^1,\ldots,x^n)\in N$ implies that
$$
\frac{\partial (f_1,\ldots,f_r)}{\partial (p_{i_1},\ldots,p_{i_r})}\neq0,
$$
for certain $i_1,\ldots,i_r\subset \{1,\ldots,n\}$, and conversely.
The expressions $f_1=\ldots=f_r=0$ can be solved implicitly for the $p_i=p_i(x^1,\ldots,x^n,\mu_1,\ldots,\mu_k)$ in terms of some functions $\mu_i=\mu_i(x^1,\ldots,x^n)$ with $i=1,\ldots,k$. Then, recall that 
$$
df=\sum_{i=1}^n\frac{\partial f}{\partial x^i}dx^i,
$$
where one can write that 
$$
\frac{\partial f}{\partial x^i}=p_i(x^1,\ldots,x^n,\mu_1,\ldots,\mu_k),\qquad i=1,\ldots,n.
$$
In order to construct a solution, one has to ensure that the $\mu_i=\mu_i(x^1,\ldots,x^n)$ are chosen such that
$d^2f=0$,
which gives us a system of partial differential equations on $\mu^1,\ldots,\mu^k$. This equation is, in general, simpler to solve using the above procedure than with the standard method, namely by using the method of characteristics successively \cite{CL11}. 

Let us apply the above method to the particular example given by the quasi-rectifiable family of vector fields (\ref{Eq:CoorSL2}) in $SL_2$. In particular, consider the vector fields
$$
Y_1^L=\frac a2\frac{\partial}{\partial a}-\frac b2 \frac{\partial }{\partial b}+\frac c2 \frac{\partial}{\partial c},\qquad Y_2^L=\frac a4\frac{\partial}{\partial a}+\left(\frac{a}{\sqrt{2}}-\frac b4\right)\frac{\partial}{\partial b}+\frac c4\frac{\partial}{\partial c}.
$$
The system of PDEs of the form $Y_1^Lf=Y_2^Lf=0$ is related to the algebraic system in the variables $p_a,p_b,p_c$ of $J^1(SL_2,SL_2\times \mathbb{R})$  given by
\begin{equation}\label{eq:system}
f_1=\frac a2p_a-\frac b2p_b+\frac c2p_c=0,\qquad f_2=\frac a4p_a+\left(\frac a{\sqrt{2}}-\frac b4\right)p_b+\frac c4p_c=0.
\end{equation}
For fixed values of $a,b,c$, it follows that $p_a,p_b,p_c$ can be written as functions $$p_a=p_a(a,b,c,\mu),\quad p_b=p_b(a,b,c,\mu),\quad p_c=p_c(a,b,c,\mu)$$ 
that depend on $a,b,c$ and a parameter $\mu$. A simple calculation shows that, for fixed $a,b,c$, all possible solutions of (\ref{eq:system}) can be written as
$$
(p_a,p_b,p_c)=\mu(-ac,0,a^2),\qquad \mu\in \mathbb{R}.
$$
It is worth noting that our coordinate system is defined on an open neighbourhood of $a=1,b=0, c=0$. To obtain a solution of $Y_1^Lf=Y_2^Lf=0$, recall that 
$$
df=\mu(a,b,c)(-ac da+a^2dc)
$$
and $d^2f=0$ for $f$ as a function of $a,b,c$. Hence, one can look for a particular parametrisation $\mu=\mu(a,b,c)$ for which $d^2f=0$. In particular, one has 
$$
d^2f=\left(\frac{\partial \mu}{\partial a}da+\frac{\partial \mu}{\partial b}db+\frac{\partial \mu}{\partial c}dc\right)\wedge (-acda+a^2dc)+3\mu a da
\wedge dc.
$$
In other words,
$$
d^2f=\frac{\partial \mu}{\partial b}(ac da\wedge db+a^2db\wedge dc)+\mu\left(a^2\frac{\partial \ln\mu}{\partial a}+ac \frac{\partial\ln \mu}{\partial c}+3a\right)da\wedge dc,
$$
which implies that $\mu=\mu(a,c)$. Then, a simple solution is, for instance, $\mu=1/a^3$. Then,
$$
df=-\frac{c}{a^2}da+\frac 1adc\,\,\Longrightarrow \,\,
f=c/a
$$
is a solution of our PDE system $Y_1^Lf=Y_2^Lf=0$.

%The evolutionary vector field for the first function is
%$$
%\mathcal{E}_{f_1}=-\frac{\partial f_1}{\partial p_i}\frac{\partial}{\partial x^i}+\frac{\partial f_1}{\partial x^i}\frac{\partial}{\partial p_i}+p_i\frac{\partial f_1}{\partial z}\frac{\partial}{\partial p_i}-p_i\frac{\partial f_1}{\partial p_i}\frac{\partial}{\partial z}=-Y^L_1+\frac {p_a}2\frac{\partial}{\partial p_a}-\frac{p_b}2\frac{\partial}{\partial p_b}+\frac{p_c}2\frac{\partial}{\partial p_c}-f\frac{\partial}{\partial z}.
%$$
%and the integrability condition gives
%$$
%\mathcal{E}_{f_1}f_2=-\frac{a p_b}{\sqrt{2}}=0,
%$$
%which gives $p_b=0$. Then, 
%$$
%f_1=f_2=0\Rightarrow p_a=-\frac cap_c
%$$
%This implies that the initial equations reduce to 
%$$a\frac{\partial f}{\partial a}+b\frac{\partial f}{\partial b}=0.
%$$
%This equation is easy to solve, giving the solution $x^3=a/c$.

As above, the same method can be applied to the vector fields (\ref{Eq:CoorSL2}), namely
$$
Y_1^L=\frac a2\frac{\partial}{\partial a}-\frac b2 \frac{\partial }{\partial b}+\frac c2 \frac{\partial}{\partial c},\qquad Y_3^L=\left(\frac b{\sqrt{2}}-\frac a2\right)\frac{\partial}{\partial a}+\frac b2 \frac{\partial }{\partial b}+\left[\frac 1{\sqrt{2}}\left(\frac{1+bc}{ab}\right)-\frac c2\right] \frac{\partial}{\partial c}
$$
or $Y_2^L,Y_3^L$.

Note that Theorem \ref{Th::RecBasis}, which has applications to the study of hydrodynamic equations \cite{Gr23}, requires the use of the Fr\"obenius theorem and the method of characteristics so as to obtain the functions $x^1,\ldots,x^n$ and then the functions $f_1,\ldots, f_r$, which are of interest to us. It is worth stressing that the integrability conditions (\ref{eq:CommRel}) ensure the existence of $f_1,\ldots,f_r$. Next, the following theorem provides an easy manner for obtaining the $f_1,\ldots,f_r$ needed to rectify the vector fields straightforwardly.

\begin{theorem}\label{Th:FroRecBett} Let $X_1,\ldots ,X_r$ be a quasi-rectifiable family of vector fields on $N$ and let $\mathcal{D}$ be the
distribution spanned by $X_1,\ldots,X_r$. Let $\eta_1,\ldots,\eta_r$ be  dual one-forms on $N$, i.e. $\eta_i(X_j)=\delta_i^j$ for $i,j=1,\ldots,r$. The nonvanishing functions $f_1,\ldots,f_r\in C^\infty(N)$ are such that  $X_1/f_1,\ldots,X_r/f_r$ commute among themselves if and only if 
\begin{equation}\label{Eq:MethodSol}
d(f_i\eta_i)|_\mathcal{D}=0,\qquad i=1,\ldots,r.
\end{equation}

\end{theorem}
\begin{proof} Let us prove the converse. Assume that $f_1,\ldots,f_r$ are such that $d(f_i
\eta_i)|_{\mathcal{D}}=0$ for $i=1,\ldots,r$. Define $Y_i={X_i}/{f_i}$ with $i=1,\ldots r$, which are vector fields dual to $f_i\eta_i$ for $i=1,\ldots,r$, namely $f_i\eta_i(X_j/f_j)=\delta_i^j$ for $i,j=1,\ldots,r$. Then, the differential of $f_i\eta_i$ vanishes on $\mathcal{D}$ by assumption and
\begin{equation}\label{eq:ExpTre}
0=d(f_i\eta_i)(Y_j,Y_k)=Y_j(\iota_{Y_k}f_i\eta_i)-Y_k(\iota_{Y_j}f_i\eta_i)-f_i\eta_i([Y_j,Y_k])=-f_i\eta_i([Y_j,Y_k]),
\end{equation}
for  $i,j,k=1,\ldots,r.$ Since the distribution $\mathcal{D}$ spanned by $X_1,\ldots,X_r$ is integrable,  $[Y_j,Y_k]$ is tangent to such a distribution. Meanwhile, (\ref{eq:ExpTre}) and the fact that $f_1,\ldots,f_r$ are non-vanishing imply that $[Y_j,Y_k]$ belongs to the annihilator of $\eta_1,\ldots,\eta_r$, which gives a supplementary distribution to $\mathcal{D}$. Hence, $[Y_j,Y_k]=0$ for $j,k=1,\ldots,r$.  

Let us prove the direct part. If $X_1/f_1,\ldots,X_r/f_r$ commute between themselves, then $d(f_i\eta_i)$ on $\mathcal{D}$ can be obtained as follows
$$
d(f_i\eta_i)(X_j/f_j,X_k/f_k)=-f_i\eta_i([X_j/f_j,X_k/f_k])=0.
$$
Hence, $d(f_i\eta_i)$ vanishes on the distribution $\mathcal{D}$ spanned by $X_1,\ldots,X_r$.
\end{proof}
Theorem \ref{Th:FroRecBett} also shows that the functions $f_1,\ldots,f_r$ are not uniquely defined since the functions needed to integrate $\eta_1,\ldots,\eta_r$, i.e. to get $d(f_i\eta_i)=0$ for $i=1,\ldots,r$ on $\mathcal{D}$, are not uniquely defined. One of the main advantages of Theorem \ref{Th:FroRecBett} in comparison with previous methods is that the functions $f_1,\ldots,f_r$ are obtained directly without finding an additional coordinate system $x^1,\ldots,x^n$ as in Theorem \ref{Th::RecBasis} and, additionally, the system of partial differential equations determining each function $f_i$ depends only on $\mathcal{D}$ and $\eta_i$.

Note that the differential forms $f_1\eta_1,\ldots,f_r\eta_r$ in Theorem \ref{Th:FroRecBett} do not need to be closed. In fact, $d(f_i\eta_i)$ only vanishes on vector fields taking values in $\mathcal{D}$, which is a condition easier to satisfy than $d(f_i\eta_i)=0$ and makes the derivation of $f_1,\ldots,f_r$ easier.

Let us apply Theorem \ref{Th:FroRecBett} to the quasi-rectifiable family of vector fields (\ref{Eq:CoorSL2}) on $SL_2$. In this case, the dual one-forms to the vector fields (\ref{Eq:CoorSL2}) are given by
$$
\begin{gathered}
\eta^L_1=\left(-\frac{2}c+\frac{(4a-\sqrt{2}b)(1+bc)}{2a^2}\right)da-\frac{1}{\sqrt{2}a}db+\left(-2b+\frac{2a^2+b^2}{\sqrt{2}a}\right)dc,\\
\eta^L_2=\frac{\sqrt{2}b(1+bc)}{a^2}da+\frac{\sqrt{2}}{a}db-\frac{\sqrt{2}b^2}{a}dc,\qquad
\eta^L_3=-\frac{2}cda+\sqrt{2}a dc.
\end{gathered}
$$
Since $Y^L_1,Y^L_2,Y^L_3$ span $TSL_2$, one has to multiply them by non-vanishing functions so that the result will become an exact differential. Then,
$$
\begin{gathered}
\frac{2}{(-2a+\sqrt{2} b)^2}\eta^L_1=d\left(\frac{c}{\sqrt{2}a}-\frac{1}{2a^2-\sqrt{2} ab}\right),\,\,
-\frac{1}{\sqrt{2}b^2}\eta^L_2=d\left(\frac{1+bc}{ab}\right),\,\,
\frac{1}{\sqrt{2} a^2}\eta^L_3=d\left(\frac{a}{c}\right).
\end{gathered}
$$
Hence, one finds, again, that the functions (\ref{Eq:RecFroFun}) allow us to rescale $Y_1^L,Y_2^L,Y_3^L$ to make them commute. Note that the previous example shows a remarkable fact: The potentials for $f_i\eta^L_i$ give us the coordinate system $x^1,x^2,x^3$ for $SL_2$ used in Theorem \ref{Th::RecBasis}. More specifically, one has the following theorem.

\begin{corollary}\label{Cor:Double} Let $X_1,\ldots,X_r$ be a quasi-rectifiable family of vector fields on $N$. Let $f_1,\ldots,f_r$ be a family of functions on $N$ such that $d(f_i\eta_i)|_{\mathcal{D}}=0$, where $\mathcal{D}$ is the distribution spanned by $X_1,\ldots,X_r$, and $\eta_1,\ldots,\eta_r$ is a dual system of one-forms to $X_1,\ldots,X_r$. If $f_i\eta_i|_{\mathcal{D}}=dx^i|_{\mathcal{D}}$ for some functions $\{x^1,\ldots,x^r\}$ with $i=1,\ldots,r$, then $X_ix^j=0$ for $i,j=1,\ldots,r$ and $i\neq j$. In other words,  $x^1,\ldots,x^r,$ along with some common functionally independent $n-r$ first integrals  for $X_1,\ldots,X_r,$ put these vector fields in quasi-rectifiable form.
\end{corollary}
\begin{proof} The proof follows from the fact that $X_jx^i=0$ for $j=1,\ldots,r$, $i=r+1,\ldots,n$. We have a family $x^{r+1},\ldots,x^n$ of functionally independent first integrals of $X_1,\ldots,X_r$ and
$$
\iota_{X_j}dx^i=f_i\iota_{X_j}\eta_i=0,\qquad i\neq j=1,\ldots,r,
$$
which hold because the $dx^1,\ldots,dx^r$ and $f_1\eta_1,\ldots,f_r\eta_r$ have the same contractions with vector fields taking values in $\mathcal{D}$.
\end{proof}

 Since $X_1,\ldots,X_r$ is a quasi-rectifiable family according to the Corollary \ref{Cor:Double}, one has that $\mathcal{D}$ becomes an integrable distribution and $f_i\eta_i|_\mathcal{D}=dx^i|_{\mathcal{D}}$, which means that only the restriction of $f_i\eta_i$ to every integral submanifold of $\mathcal{D}$ has to be exact. An example of this fact is to be presented in Section \ref{Se:Hydro} so as to illustrate the relevance of our method and to study the sound Lie algebras of vector fields related to the propagation of sound waves occurring in $(1+1)$-dimensional hydrodynamic-type equations. 

There is another structure that appears in the practical cases analysed in the following sections. This structure will lead to a system of partial differential equations determining the functions $f_1,\ldots,f_r$ in Corollary \ref{Cor:Double}. Assume that the vector fields $X_1,\ldots,X_r$ on the manifold $N$ can be extended to a family $X_1,\ldots,X_n$ of vector fields such that $X_1\wedge\ldots\wedge X_n$ does not vanish on $N$. Then, $X_1,\ldots,X_n$ span the  tangent bundle $TN$. The form of the extended vector fields is not really important, but it will be related to quasi-rectifiable families of vector fields in practical cases. Under the above assumptions, there exist dual forms $\eta_1,\ldots,\eta_n$ to $X_1,\ldots,X_n$. Hence, one can calculate the differentials of the one-forms as follows
$$
d\eta_i(X_j,X_k)=X_j\eta_i(X_k)-X_k\eta_i(X_j)-\eta_i([X_j,X_k])=-\eta_i([X_j,X_k])
$$
for $i,j,k=1,\ldots,n$. If we write $[X_j,X_k]=\sum_{i=1}^n f_{jk}^iX_i$ for some uniquely defined functions $f_{jk}^i\in C^\infty(N)$ and $j,k=1,\ldots,n$, then
$$
d\eta_i=-\frac 12\sum_{j,k=1}^nf^i_{jk}\eta_j\wedge \eta_k,\qquad i=1,\ldots,r.
$$
Then,
$$
d(f_i\eta_i)=df_i\wedge \eta_i-\frac {f_i}2\sum_{j,k=1}^nf^i_{jk}\eta_j\wedge \eta_k,\qquad i=1,\ldots,r.
$$
If $d(f_i\eta_i)|_\mathcal{D}=0$ for $i=1,\ldots,r$, then
$$
df_i|_{\mathcal{D}}\wedge\eta_i-\frac {f_i}2\sum_{j,k=1}^rf^i_{jk}\eta_j\wedge \eta_k=0,\qquad i=1,\ldots,r.
$$
Since the chosen family $X_1,\ldots,X_r$ is quasi-rectifiable,
$$
\left(df_i|_{\mathcal{D}}+ {f_i}\sum_{k=1}^rf^i_{ik}\eta_k\right)\wedge \eta_i=0,\qquad i=1,\ldots,r.
$$
holds, and the equations determining each $f_i$ are
$$
X_j\ln|f_i|=-f^i_{ij},\qquad j=1,\ldots,r, \qquad j\neq i.
$$
Since $X_1,\ldots,X_r$ can be put in quasi-rectifiable form, it can be proved that the above system always admits a solution.

Let us finally describe in detail a new method that can be of some interest in certain circumstances. More specifically, we are now interested in finding integrability conditions for systems of PDEs of the form
\begin{equation}\label{eq:DifEq2}
X_{(i)}f=g_i,\qquad i=1,\ldots,n,
\end{equation}
for several functions $g_1,\ldots,g_n:N\times \mathbb{R}\rightarrow\mathbb{R}$ depending on $N$ and $f$, and a family of vector fields $X_{(1)},\ldots,X_{(n)}$ which spans the tangent bundle $TN$. For instance, (\ref{eq:DifEq2}) is interesting when $g_2=\ldots=g_n=0$, as systems of PDEs of this type lead us to put $X_{(1)},\ldots,X_{(n)}$ into a quasi-rectifiable form. Moreover, systems of PDEs of the form (\ref{eq:DifEq2}) occur very frequently in the literature. As a particular instance, we  generalise and understand geometrically  the results of \cite[pg. 91]{Sn06} for a particular class of systems (\ref{eq:DifEq2}) on $\mathbb{R}^2$. In particular, we will provide a new application of the so-called {\it evolution vector fields} \cite{SLLM20}.
 The evolutionary vector field in $J^1
 \pi$ related to a function $f\in C^\infty(J^1\pi)$  takes the form (see \cite{LLM21,SLLM20} for further details\footnote{There is a typo in the second line of the equations of motion for an evolution vector field \cite[p.  6]{SLLM20}, where $\partial H/\partial p^i$ should be $\partial H/\partial q^i$ as in \cite[p. 2]{LLM21}.})
$$
\mathcal{E}_{f}=-\sum_{i=1}^n\left(\frac{\partial f }{\partial x^i}\frac{\partial}{\partial p_i}-\frac{\partial f }{\partial p_i}\frac{\partial}{\partial x^i}+p_i\frac{\partial f }{\partial z}\frac{\partial}{\partial p_i}-p_i\frac{\partial f }{\partial p_i}\frac{\partial}{\partial z}\right).
$$

\begin{proposition}\label{Prop:NewMethod} Let $X_{(1)},\ldots,X_{(n)}$ be a family of vector fields on an $n$-dimensional manifold $N$ spanning its tangent bundle around $x\in N$. Let $\{x^1,\ldots,x^n,z,p_1,\ldots, p_n\}$ be a locally adapted coordinate system for $J^1\pi$ and define 
\begin{equation}\label{Def:fun2}
f_j(x^1,\ldots,x^n,z,p_1,\ldots,p_n)=\sum_{i=1}^nX_{(j)}^i(x^1,\ldots,x^n)p_i-g_j(x^1,\ldots,x^n,z),\qquad j=1,\ldots,n,
\end{equation}
where 
$$
X_{(j)}=\sum_{i=1}^nX_{(j)}^i(x^1,\ldots,x^n)\frac{\partial}{\partial x^i},\qquad j=1,\ldots,n.
$$
If a system of partial differential equations on $N$ of the form
(\ref{eq:DifEq2}) admits a solution $f\in C^\infty(U)$ on a neighbourhood $U$ of $x$, then the equations (\ref{eq:DifEq2}), considered as a system of  equations $f_1=0,\ldots,f_n=0$ in $J^1\pi$, satisfy, on the lift $j^1\sigma_f$ of  $x\in N\mapsto (x,f(z))\in N\times \mathbb{R}$ to $J^1\pi$, the condition that
$$
(j^1\sigma_f)^*\{f_1,\ldots,f_n\}_{ij}=0, \qquad 1\leq i<j\leq n,
$$
for a series of $n(n-1)/2$ brackets $\{\cdot,\ldots,\cdot\}_{ij}:C^\infty(J^1\pi)^n\rightarrow C^\infty(J^1\pi)$ that are derivations on each entry. If $n=2$, then the above expression reduces to
\begin{equation}\label{eq:SnCh}
\mathcal{E}_{f_1}f_2|_{j^1\sigma_f}=0.
\end{equation}
\end{proposition}
\begin{proof}
 In the adapted coordinates $\{x^1,\ldots,x^n,z,p_1,\ldots,p_n\}$ of $J^1\pi$, the system (\ref{eq:DifEq2}) can be rewritten as follows
$$
f_i\left(x^1,\ldots,x^n,f(x^1,\ldots,x^n),\frac{\partial f}{\partial x^1}(x^1,\ldots,x^n),\ldots,\frac{\partial f}{\partial x^n}(x^1,\ldots,x^n)\right)=0,\qquad i=1,\ldots,n.
$$
The fact that $X^1\wedge\ldots\wedge X^n$ does not vanish in a neighbourhood of $x\in N$ implies that
\begin{equation}\label{eq:Nond}
\frac{\partial (f_1,\ldots,f_n)}{\partial (p_{1},\ldots,p_{n})}\neq0
\end{equation}
and conversely. % Let us assume, without loss of generality, that the $(r\times r)$-matrix with coefficients $X^i_{(j)}(x^1,\ldots,x^n)$, with $i,j=1,
%\ldots,r$, is non-degenerate on a neighbourhood of a point $x\in N$. 
 A solution $f\in C^\infty(N)$ of (\ref{eq:DifEq2}) gives rise to a section $\sigma_f(x)=(x,f(x))$ of $\pi:N\times \mathbb{R}\rightarrow N$, which, in turn, leads to a lift $j^1\sigma_f$ of $\sigma_f$ to $J^1\pi$ given by
$$
j^1\sigma_f(x^1,\ldots,x^n)=\left(x^1,\ldots,x^n,f(x^1,\ldots,x^n),\frac{\partial f}{\partial x^1}(x^1,\ldots,x^n),\ldots,\frac{\partial f}{\partial x^n}(x^1,\ldots,x^n)\right).
$$
To characterise lifts, one may use the contact form $\alpha=dz-\sum_{i=1}^np_idx^i$ on $J^1\pi$. Then, a section 
$$
\sigma(x^1,\ldots,x^n)=(x^1,\ldots,x^n,z(x^1,\ldots,x^n),p_1(x^1,\ldots,x^n),\ldots,p_n(x^1,\ldots,x^n))
$$
of $\pi^1:(x,z,p)\in J^1\pi\mapsto x\in N$ satisfying
$\sigma^*\alpha=0$ gives
$$
p_i=\frac{\partial f}{\partial x^i}(x^1,\ldots,x^n),\qquad i=1,\ldots,n.
$$
The condition (\ref{eq:Nond}) allows us to write $p_i=\phi_i(x,z)$ for certain functions $\phi_1,\ldots,\phi_n:\mathbb{R}^{n+1}\rightarrow \mathbb{R}$. It is worth stressing that (\ref{Def:fun2}) shows that (\ref{eq:DifEq2}) can be considered as a linear system of equations with respect to $p_1,\ldots,p_n$. The condition (\ref{eq:Nond}) implies that its matrix of coefficients of $p_1,\ldots,p_n$ is invertible. As such, one can describe solutions for $p_1,\ldots,p_n$ in terms of the coefficients of the system via Cramer's method, and the obtained expressions depend only on $x^1,\ldots,x^n$, and, possibly $z$. Then,
\begin{equation}\label{eq:Expr}
f_l(x,z,\phi_i(x,z))=0,\qquad l=1,\ldots,n.
\end{equation}

Using the relations $p_i=\phi_i(x^1,\ldots,x^n,z)$ and considering particular solutions $z=z(x)$, one obtains
\begin{equation}\label{Eq:Sym}
\frac{\partial^2 z}{\partial x^j\partial x^i}=\frac{\partial p_i}{\partial x^j}=\frac{\partial \phi_i}{\partial x^j}+\frac{\partial \phi_i}{\partial z}\phi_j=\frac{\partial^2 z}{\partial x^i\partial x^j}=\frac{\partial p_j}{\partial x^i}=\frac{\partial \phi_j}{\partial x^i}+\frac{\partial \phi_j}{\partial z}\phi_i,\qquad 1\leq i<j\leq n.
\end{equation}
Meanwhile, the partial derivatives of the equations (\ref{eq:Expr}) in terms of $x^1,\ldots,x^n,z$ are
$$
\frac{\partial f_l}{\partial x^j}+\sum_{i=1}^n\frac{\partial f_l}{\partial p_i}\frac{\partial\phi_i}{\partial x^j}=0,\qquad \frac{\partial f_l}{\partial z}+\sum_{i=1}^n\frac{\partial f_l}{\partial p_i}\frac{\partial \phi_i}{\partial z}=0,\qquad l,j=1,\ldots,n.
$$
The above implies a series of relations given by the matrix equation
$$
M+FZ=0,
$$
where
$$
Z_{ij}=\frac{\partial \phi_i}{\partial x^j}+\phi_j\frac{\partial \phi_i}{\partial z},\qquad M_{lj}=\frac{\partial f_l}{\partial x^j}+\phi_j\frac{\partial f_l}{\partial z},\qquad F_{l}^j=\frac{\partial f_l}{\partial p_j},\qquad i,j,l=1,\ldots,n.
$$
Again (\ref{eq:Nond}) ensures that  $F$ admits an inverse $F^{-1}$ and one can write
$$
Z=-F^{-1}M.
$$
As $Z$ is a symmetric matrix due to conditions (\ref{Eq:Sym}), it follows that 
$$
F^{-1}M=M^T(F^{-1})^T.
$$
Since $F^{-1}=({\det F})^{-1}{\rm adj}F$, where adj$F$ is the adjoint matrix of $F$, one has that 
\begin{equation}\label{eq:Com}
{\rm adj}(F)M=M^T{\rm adj}(F^T).
\end{equation}
The entries of ${\rm adj}(F)$ are minors of $F$, which implies that they are homogeneous polynomials of order $n-1$ in the partial derivatives of the $f_i$ with respect to the momenta $p_j$. In particular, if $I^{n-1}_{\alpha,j}$ is any $(n-1)$-index $\alpha=(j_1,\ldots,j_{n-1})$ where $j_1,\ldots,j_{n-1}$ are different numbers contained in $\{1,\ldots,\widehat{j},\ldots, n\}$, then 

\begin{align*}
{\rm adj}(F)_{i}^j&=\sum_{I^{n-1}_{\alpha,j}}(-1)^{i+j}\epsilon_{j_1,\ldots,j_{n-1}}\frac{\partial f_{j_1}}{\partial p_{1}}\cdots\widehat{\frac{\partial f_{j}}{\partial p_{i}}}\cdots \frac{\partial f_{j_{n-1}}}{\partial p_{n}}
\\&=\sum_{I^{n-1}_{\alpha,j}}(-1)^{i+j}\epsilon_{j_1,\ldots,j_{n-1}}X^1_{({i_1})}\cdots\widehat{X^i_{({j})}}\cdots X^n_{(i_{n-1})}.
\end{align*}
It is worth noting that 
$$
{\rm adj}(F)_{i}^j=(-1)^{i+j}\frac{\partial(f_1,\ldots,\widehat{f}_j,\ldots ,f_n)}{\partial (p_1,\ldots,\widehat{p_i},\ldots,p_n)},
$$
where the determinant on the right-hand side is, by definition, the {\it Nambu bracket}, of $f_1,\ldots,\widehat{f}_j,\ldots,f_n$ in terms of the variables $p_1,\ldots,\widehat{p_i},\ldots,p_n$ (cf. \cite{Sh06}). 
Using these expressions and  (\ref{eq:Com}), one gets
\begin{multline*}
\sum_{k=1}^n\sum_{I^{n-1}_{\alpha,k}}(-1)^{i+k}\epsilon_{j_1,\ldots,j_{n-1}}X^{1}_{(i_1)}\ldots \widehat{X^{i}_{(k)}}\ldots X_{(i_{n-1})}^n\left(\frac{\partial f_k}{\partial x^j}+\phi_j\frac{\partial f_k}{\partial z}\right)\\=\sum_{k=1}^n
\left(\frac{\partial f_k}{\partial x^i}+\phi_i\frac{\partial f_k}{\partial z}\right)\sum_{I^{n-1}_{\alpha,k}}(-1)^{j+k}\epsilon_{j_1,\ldots,j_{n-1}}X^{1}_{(i_1)}\ldots \widehat{X^{j}_{(k)}}\ldots X_{(i_{n-1})}^n
\end{multline*}
for $i,j=1,\ldots,n$. These relations are derivations on each $f_1,\ldots,f_n$ and can therefore be described by means of  $n$-vector fields $\Lambda_{ij}$ for $1\leq i<j\leq n$ evaluated when $p_i=\phi_i$ for $i=1,\ldots,n$. In particular, if $n=2$, one obtains a single expression that can be described via the evolutionary vector field of $f_i$,  say $\mathcal{E}_{f_i}$, which is
$$
\mathcal{E}_{f_i}=-\sum_{j=1}^n\left(\frac{\partial f_i}{\partial x^j}\frac{\partial}{\partial p_j}-\frac{\partial f_i}{\partial p_j}\frac{\partial}{\partial x^j}+p_j\frac{\partial f_i}{\partial z}\frac{\partial}{\partial p_j}-p_j\frac{\partial f_i}{\partial p_j}\frac{\partial}{\partial z}\right).
$$
We then obtain
$$
\{f_1,f_2\}=\mathcal{E}_{f_1}f_2
$$
that vanishes on $j^1\sigma_f$. This is  the integrability condition for solutions of our initial system. 

\end{proof}

It is worth noting that the previous condition  (\ref{eq:SnCh}) for $n=2$ was used in \cite{Sn06} to solve systems of partial differential equations on $\mathbb{R}^2$. Here we provide a modern geometric approach to the topic. 
\section{Integrable systems arising from  quasi-rectifiable families of vector fields}\label{Sec::Int}

Let us investigate the relevance of quasi-rectifiable families of vector fields in the study of the integrability of systems of first-order differential equations \cite{LS93}. Consider a system of ordinary differential equations on a manifold $N$ of the form
\begin{equation}\label{eq:Initial}
\frac{dx^i}{dt}=X^i(x),\qquad i=1,\ldots,n.
\end{equation}
This system determines a vector field $X=X^i(x)\partial/\partial x^i$ on $N$, which describes its integral curves, and vice versa. 

One of the standard methods for studying (\ref{eq:Initial}) is based on the use of Lie symmetries of $X$, i.e. vector fields $Y$ on $N$ such that $[X,Y]=0$. Then, the elements of the group of diffeomorphisms related to $Y$ map solutions of $X$ into solutions of $X$. This allows for the simplification and analysis of the properties of  (\ref{eq:Initial}) (see \cite{Ol83}).

Assume that $X$ forms part of an almost  rectifiable family of vector fields $X_1=X,X_2,\ldots,X_r$. Then, these vector fields can be multiplied by the non-vanishing functions $f_1,\ldots,f_r$, respectively, so that $[f_iX_i,f_jX_j]=0$ for $i,j=1,\ldots,r$. In particular,  $f_1X$ is related to the new system of ordinary differential equations
\begin{equation}\label{eq:Trans}
\frac{d\hat x^i}{d\tau}=f_1X^i(\hat x),\qquad i=1,\ldots,n.
\end{equation}
Note that $\hat x(\tau )$ is a solution of (\ref{eq:Trans}) if and only if the $t$-reparametrisation 
$$
\tau(t)=\int_0^t \frac{dt'}{f_1(x(t'))}
$$
is such that $
\hat{x}(\tau(t))=x(t)
$ is a solution of (\ref{eq:Initial}). This transformation is called a {\it Sundman transformation} and there has been much interest in it \cite{CMM22,CMM23,EE04}. Note that, for the transformed system (\ref{eq:Trans}), the vector fields $f_2X_2,\ldots,f_rX_r$ are Lie symmetries of $f_1X$, which can be used to integrate the system (\ref{eq:Trans}) and to study its solutions. Once the solutions of (\ref{eq:Trans}) have been obtained, one can retrieve the solutions of (\ref{eq:Initial}) by writing $x(t)=\hat{x}(\tau(t))$, for each particular solution $
\hat{x}(\tau)$ of (\ref{eq:Trans}).

The vector fields $f_2X_2,\ldots,f_rX_r$ span an $(r-1)$-dimensional  Abelian Lie algebra of vector fields. They can be integrated in order to define a Lie group action $\varphi:\mathbb{R}^{n-1}\times N\rightarrow N$ of symmetries of $f_1X_1$ (see \cite{LS93}). This Lie group action is not a Lie group action of symmetries of \eqref{eq:Initial}, but the transformations map solutions into particular solutions up to a parametrisation. In this case, we say that $\varphi:\mathbb{R}^{n-1}\times N\rightarrow N$ is a Lie group of Sundman symmetries of \eqref{eq:Initial}. Let us give a formal general definition.

\begin{definition} Given a system of first-order differential equations (\ref{eq:Initial}) on $N$, we define a Lie group of Sundman transformations to be a Lie group action $\varphi:G\times N\rightarrow N$ mapping solutions of (\ref{eq:Initial}) into solutions of (\ref{eq:Initial}) up to time-reparametrisations.
\end{definition}

Note that the time-reparametrisations in the above definition may be different for each particular solution of (\ref{eq:Initial}).

Let us now study a more specific type of system (\ref{eq:Initial}), in particular, those that are Hamiltonian relative to a symplectic form. We aim to briefly analyse the relation of these systems with quasi-rectifiable Lie algebras of vector fields and integrable systems in a symplectic Hamiltonian form.

\begin{theorem}\label{Th:HamRec} A family of vector fields $X_1,\ldots,X_r$ on $N$ is a quasi-rectifiable family of Hamiltonian vector fields relative to a symplectic form $\omega$ on $N$ if and only if there exists a family of Hamiltonian functions $h_1,\ldots,h_r$ on $N$ for the vector fields $X_1,\ldots,X_r$, respectively, such that the Poisson bracket of $h_i,h_j$ is of the form $\{h_i,h_j\}=h_{ij}(h_i,h_j)$ for some functions $h_{ij}\in C^\infty(\mathbb{R}^2)$ with $1\leq i<j\leq r$.
\end{theorem}
\begin{proof} Since $X_1,\ldots,X_r$ form a quasi-rectifiable family of vector fields, one has  $[X_i,X_j]=f_{ij}^iX_i+f_{ij}^jX_j$ for $1\leq i<j\leq r$ for certain functions $f_{ij}^i,f_{ij}^j\in C^\infty(N)$. If, in addition, $X_1,\ldots,X_r$ are Hamiltonian vector fields relative to $\omega$, then the commutator of Hamiltonian vector fields is Hamiltonian. Hence, each $f_{ij}^iX_i+f_{ij}^jX_j$ is a Hamiltonian vector field and it admits a certain Hamiltonian function $\Upsilon_{ij}$, i.e. $\iota_{f_{ij}^iX_i+f_{ij}^jX_j}
\omega=d\Upsilon_{ij}$. Then,
$$
d^2\Upsilon_{ij}=d\iota_{f_{ij}^iX_i+f_{ij}^jX_j}\omega=0,\qquad 1\leq i<j\leq r.
$$
Hence,
\begin{equation}\label{Eq:Con}
0=d(f_{ij}^idh_i+f_{ij}^jdh_j)=df_{ij}^i\wedge dh_i+df_{ij}^j\wedge dh_j,\qquad 1\leq i<j\leq r.
\end{equation}
Since $X_1\wedge\ldots\wedge X_r$ is not vanishing and the mapping $\omega^\flat:v\in TN\mapsto \omega(v,\cdot)\in T^*N$ is an isomorphism because $\omega$ is symplectic and therefore nondegenerate, one has that $dh_1\wedge\ldots\wedge dh_r\neq 0$. Then, $h_1,\ldots,h_r$ are functionally independent functions and some additional functions $y_1,\ldots,y_s$ can be added to them so as to obtain a coordinate system on $T^*N$. Using this coordinate system in (\ref{Eq:Con}), it follows that
$$
\sum_{k=1}^r\frac{\partial f_{ij}^i}{\partial h_k}dh_k\wedge dh_i+\sum_{l=1}^s\frac{\partial f_{ij}^i}{\partial y_l}dy_l\wedge dh_i+\sum_{k=1}^r\frac{\partial f_{ij}^j}{\partial h_k}dh_k\wedge dh_j+\sum_{l=1}^s\frac{\partial f_{ij}^j}{\partial y_l}dy_l\wedge dh_j=0,\quad 1\leq i<j\leq r.
$$
Since $i\neq j$, the linear independence of the basis $dh_1,\ldots,dh_r,dy_1,\ldots,dy_s$ allows one to write
$$
\begin{gathered}
\frac{\partial f_{ij}^i}{\partial h_k}dh_k\wedge dh_i=0,\quad \frac{\partial f_{ij}^i}{\partial y_l}dy_l\wedge dh_i=0,\quad \frac{\partial f_{ij}^j}{\partial y_l}dy_l\wedge dh_j=0,\quad \frac{\partial f_{ij}^j}{\partial h_k}dh_k\wedge dh_j=0,\,\,\,k\notin \{i,j\}\\
\left(\frac{\partial f_{ij}^i}{\partial h_j}-\frac{\partial f_{ij}^j}{\partial h_i}\right)dh_i\wedge dh_j=0.
\end{gathered}
$$
The four equalities in the first line above give that $f_{ij}^i=f_{ij}^i(h_i,h_j)$ and $f_{ij}^j=f_{ij}^j(h_i,h_j)$. Moreover, the second line above yields
$$
\frac{\partial f_{ij}^i}{\partial h_j}=\frac{\partial f_{ij}^j}{\partial h_i},\qquad 1\leq i<j\leq r.
$$
Consequently, there exists a series of functions $h_{ij}=h_{ij}(h_i,h_j)$, with $1\leq i<j\leq r$, such that
$$
f_{ij}^i=\frac{\partial h_{ij}}{\partial h_i},\qquad f_{ij}^j=\frac{\partial h_{ij}}{\partial h_j},\qquad 1\leq i<j\leq r.
$$
Let us prove the converse. If $\{h_i,h_j\}=h_{ij}(h_i,h_j)$ for $1\leq i,j\leq r$, then
$$
d\{h_i,h_j\}=\frac{\partial h_{ij}}{\partial h_i}dh_{i}+\frac{\partial h_{ij}}{\partial h_j}dh_{j},\qquad 1\leq i<j\leq r.
$$
Since $d\{h_i,h_j\},h_i,h_j$ are the Hamiltonian functions for $-[X_i,X_j],X_i,X_j$, respectively, the above expression implies that
$$
-[X_i,X_j]=\frac{\partial h_{ij}}{\partial h_i}X_{i}+\frac{\partial h_{ij}}{\partial h_j}X_{j},\qquad 1\leq i<j\leq r,
$$
and the vector fields $X_1,\ldots,X_r$ form a quasi-rectifiable family.

\end{proof}
The above theorem justifies the following definition.
\begin{definition}A family of functions $h_1,\ldots,h_r$ on a symplectic manifold $(N,\omega)$ is {\it quasi-rectifiable} if  there exist functions $h_{ij}(h_i,h_j)$, with $1\leq i<j\leq r$, such that
$$
\{h_i,h_j\}=h_{ij}(h_i,h_j),\qquad 1\leq i<j\leq r.
$$    
\end{definition}

The above definition covers, as a particular case, the functions defining a completely integrable or  superintegrable Hamiltonian system. In fact, in this case, one has a series of functions $h_1,\ldots,h_r$ such that $\{h_i,h_j\}=0$ for $1\leq i<j\leq r$ and $2r=\dim N$ for the integrable case, or $r>\dim N/2$ for the superintegrable one. Expressions of the above type may also occur in the theory of deformation of Hamiltonian systems with Poisson bialgebras introduced in \cite{BCFHL17} and  developed further in \cite{FCH23}. 

Quasi-rectifiable families of Hamiltonian vector fields cannot, in general, become families of commuting Hamiltonian vector fields by rescalings. Some additional conditions must be imposed on the functions $h_{ij}$ in order to ensure this. The following proposition analyses necessary and sufficient conditions for $X_1,\ldots,X_r$ to form a quasi-rectifiable family of Hamiltonian vector fields that can be rescaled to commuting Hamiltonian vector fields.

\begin{proposition} Let $X_1,\ldots,X_r$ be a quasi-rectifiable family of Hamiltonian vector fields on a manifold $N$ relative to a symplectic form $\omega$ with Hamiltonian functions $h_1,\ldots,h_r$, respectively. Then, $X_1,\ldots,X_r$ can be rescaled into a family of Hamiltonian commuting vector fields if and only if the non-vanishing Poisson bracket between their Hamiltonian functions is $\{h_i,h_j\}=H_i(h_i)H_j(h_j)$ for $1\leq i<j\leq r$  and some functions $H_1,\ldots,H_r\in C^\infty(\mathbb{R})$.
\end{proposition}
\begin{proof} Let $h_1,h_2$ be Hamiltonian functions for $X_1,X_2$. As $X_1,X_2$ are quasi-rectifiable, Proposition \ref{Th:HamRec} shows that
$\{h_1,h_2\}=h_{12}(h_1,h_2)$ for some function $h_{12}\in C^\infty(\mathbb{R}^2)$. Then,
$$
[X_1,X_2]=-\frac{\partial h_{12}}{\partial h_1}X_1-\frac{\partial h_{12}}{\partial h_2}X_2.
$$
Note that a vector field $f_iX_i$, with $i=1,\ldots,r$, is again a Hamiltonian vector field if and only if $f_i=f_i(h_i)$. In fact,
$$
d\iota_{f_iX_i}\omega=df_i\wedge dh_i
$$
is equal to zero if and only if $f_i=f_i(h_i)$. Using this result, let us rectify $X_1,X_2$ by rescaling via two non-vanishing functions $f_1,f_2$, namely
$$
0=[f_1X_1,f_2X_2]=-f_2(X_2f_1)X_1+f_1(X_1f_2)X_2-f_1f_2\frac{\partial h_{12}}{\partial h_1}X_1-f_1f_2\frac{\partial h_{12}}{\partial h_2}X_2.
$$
Since $X_2f_1=\{f_1,h_2\}=(\partial_{h_1}f_1)\{h_1,h_2\}=(\partial_{h_1}f_1)h_{12}$ and $f_1,f_2$ do not vanish, it follows that
$$
-\frac{\partial f_1}{\partial h_1}h_{12}-f_1\frac{\partial h_{12}}{\partial h_1}=0,\qquad -\frac{\partial f_2}{\partial h_2}h_{12}-f_2\frac{\partial h_{12}}{\partial h_2}=0.
$$
If $h_{12}\neq 0$, and assuming without loss of generality that $f_1,f_2>0$, one gets that
\begin{equation}\label{eq:Step}
\frac{\partial \ln f_1}{\partial h_1}=-\frac{\partial \ln h_{12}}{\partial h_1},\qquad \frac{\partial\ln f_2}{\partial h_2}=-\frac{\partial \ln h_{12}}{\partial h_2}.
\end{equation}
Since the left-hand sides of the above equations depend only on $h_1$ and $h_2$, respectively, $\ln h_{12}$ has to be a linear combination of two functions depending on $
h_1$ and $h_2$. Hence,
$$
\ln h_{12}=F_1(h_1)+F_2(h_2)
$$
for some functions $F_1,F_2\in C^\infty(\mathbb{R})$. Hence, $h_{12}=H_1(h_1)H_2(h_2)$ for some functions $H_1,H_2$ if $
\{h_1,h_2
\}\neq 0$. If $h_{12}=0$, the decomposition still holds. The same applies to all  of the remaining commutators $\{h_i,h_j\}$, with $1\leq i<j\leq r$.

The proof of the converse statement is an immediate consequence of (\ref{eq:Step}).
\end{proof}

 \section{Analysis of quasi-rectifiable Lie algebras}\label{Sec::ClaRecLieAlg}

This section defines and analyses quasi-rectifiable Lie algebras not necessarily related to Lie algebras of vector fields. The practical relevance of this concept and their applications will be developed in Section \ref{Se:Hydro}. In particular, we will show that quasi-rectifiable Lie algebras appear naturally in the analysis of hydrodynamic-type equations by means of Riemann invariants  \cite{Gr23,GH07}. Moreover, quasi-rectifiable Lie algebras will also be related to the so-called PDE Lie systems and their applications to hydrodynamic-type equations \cite{CGM07}.

\subsection{Definition and general properties}
Let us  define and prove some general results concerning quasi-rectifiable Lie algebras.

\begin{definition} A Lie algebra $\mathfrak{g}$  is \textit{quasi-rectifiable} if it admits a basis $
\{e_1,\ldots,e_r\}$ such that the Lie bracket of any two elements of the basis belongs to the linear space spanned by them. The basis $\{e_1,\ldots,e_r\}$ is called {\it quasi-rectifiable}. If $\mathfrak{g}$ admits no quasi-rectifiable basis, $\mathfrak{g}$ is called a {\it non quasi-rectifiable algebra}.  
\end{definition}

It is worth noting that quasi-rectifiable Lie algebras are defined in terms of a basis-dependent condition. Indeed, if a basis of a Lie algebra $\mathfrak{g}$ is quasi-rectifiable, it is possible that another basis of $\mathfrak{g}$ will not be so. It is therefore relevant to characterise algebraically/geometrically when a Lie algebra admits a quasi-rectifiable basis.  

Recall, for instance, that we proved in (\ref{Eq:basis_sl2}) that the Lie algebra $\mathfrak{sl}_2$ admits a basis  $\{e_1,e_2,e_3\}$ such that
$$
[e_1,e_2]=e_2-\frac 12 e_1,\qquad, [e_1,e_3]=-e_3-e_1,\qquad [e_2,e_3]=e_2-\frac 12 e_3.  
$$
Then, $\mathfrak{sl}_2$ is a quasi-rectifiable Lie algebra and this basis is a quasi-rectifiable one. Nevertheless, it is known that bases for $\mathfrak{sl}_2$ are more frequently written so that they admit structure constants of the form (\ref{ConRel}), which show that such bases are not quasi-rectifiable.

Note that the Lie algebra  $\mathfrak{so}_3$ of generators of rotations in the three-dimensional space is isomorphic to $\mathbb{R}^3$ endowed with the cross product $\times$. It follows that $\mathbb{R}^3$ is not a quasi-rectifiable Lie algebra relative to the cross product. In fact, the cross product of any two linearly independent elements in $\mathbb{R}^3$ is a third, linearly independent, one. The isomorphism between $\mathfrak{so}_3$ and $\mathbb{R}^3$ shows that $\mathfrak{so}_3$ is not quasi-rectifiable either.

Let us use the following method to characterise all quasi-rectifiable Lie algebras. The procedure is based on the description of an algebraic equation whose solutions may give rise to quasi-rectifiable bases. 

\begin{theorem}\label{Th::RectCri} A Lie algebra $\mathfrak{g}$ is quasi-rectifiable if and only if there exists a dual basis $\{e^1,\ldots,e^r\}$ of elements in $\mathfrak{g}^*$ such that 
\begin{equation}\label{Eq:ConRec}
e^i\wedge \delta e^i=0,\qquad i=1,\ldots,r,
\end{equation}
where $
\delta:\mathfrak{g}^*\rightarrow \mathfrak{g}^*\wedge \mathfrak{g}^*$ is equal to minus the transpose of the Lie bracket $[\cdot,\cdot]:v_1\wedge v_2\in \mathfrak{g}\wedge  \mathfrak{g}\mapsto [v_1,v_2]\in \mathfrak{g}$, i.e. 
\begin{equation}\label{eq::Trans}
\delta=-[\cdot,\cdot]^T.
\end{equation}
\end{theorem}
\begin{proof} If $\{e_1,\ldots,e_r\}$ is a quasi-rectifiable basis of $\mathfrak{g}$, then
$$
[e_i,e_j]=c_{ij}^ie_i+c_{ij}^je_j,\qquad 1\leq i<j\leq r.
$$
It is worth noting that we understand the Lie bracket $[\cdot,\cdot]:\mathfrak{g}\times\mathfrak{g}\rightarrow \mathfrak{g}$, which is bilinear, as a linear mapping $[\cdot,\cdot]:v_1\wedge v_2\in \mathfrak{g}\wedge\mathfrak{g}\mapsto [v_1,v_2]\in \mathfrak{g}$. Hence, given the dual basis $\{e^1,\ldots,e^r\}$ in $\mathfrak{g}^*$ and (\ref{eq::Trans}), we have 
$$
\delta e^i(e_j\wedge e_k)=-e^i([e_j,e_k])=-e^i(c_{jk}^je_j+c_{jk}^ke_k)=-\delta^i_jc_{jk}^j-c_{jk}^k\delta_k^i,\qquad 1\leq j<k\leq r.
$$
Hence, one can write
$$
\delta e^i=-\sum_{j=1}^re^i\wedge c_{ij}^ie^j,\qquad i=1,\ldots,r,
$$
which proves the direct part. 

The proof of the converse is as follows. By the assumption (\ref{Eq:ConRec}), there exist some covectors $\vartheta^1,\ldots,\vartheta^r\in \mathfrak{g}^*$ such that
$$
\delta e^i=e^i\wedge \vartheta^i,\qquad i=1,\ldots,r.
$$
  Hence, given the dual basis $\{e_1,\ldots,e_r\}$ to $\{e^1,\ldots,e^r\}$, one has
  $$
  e^i([e_j,e_k])=-\delta e^i(e_j,e_k)=-e^i\wedge \vartheta^i(e_j,e_k),\qquad i,j,k=1,\ldots,r,
  $$
  which means that the result will zero if $e_j$ and $e_k$ are different from $e_i$. Repeating this for the basis $\{e^1,\ldots,e^r\}$, one obtains that the Lie bracket $[e_j,e_k]$ is a linear combination of  $e_j$ and $e_k$. 
\end{proof}

Theorem \ref{Th::RectCri} shows that $\delta$ allows us to characterise whether a quasi-rectifiable basis is available. It is worth noting that the equation $\vartheta\wedge \delta \vartheta \neq 0$, with $\vartheta\in \mathfrak{g}^*$, can be used to characterise certain three-dimensional contact structures on Lie groups \cite{LR22}. Its similarity with its relation to (\ref{Eq:ConRec}) is immediate.

There are many algebraic results on quasi-rectifiable Lie algebras that can easily be obtained.  The following propositions are straightforward to prove.

\begin{proposition} Every quasi-rectifiable $r$-dimensional Lie algebra has Lie subalgebras of dimensions from zero to $r$.
\end{proposition}
\begin{proof}
Choose a quasi-rectifiable basis $\{e_1,\ldots,e_r\}$ of the Lie algebra and define the basis
$$
\{0,\langle e_1\rangle, \langle e_1,e_2\rangle,\ldots,\langle e_1,\ldots,e_r\rangle\}.
$$\end{proof}

\begin{proposition} The direct sum of quasi-rectifiable Lie algebras is quasi-rectifiable.
\end{proposition}

More interestingly, one has the following result.

\begin{proposition} If $\rho:\mathfrak{g}\rightarrow \mathfrak{h}$ is a surjective Lie algebra morphism and $\mathfrak{g}$ is a quasi-rectifiable Lie algebra, then $\mathfrak{h}$ is quasi-rectifiable.
\end{proposition}
\begin{proof} This proposition is due to the fact that there exists a quasi-rectifiable basis $\{e^1,\ldots, e^r\}$ of $\mathfrak{g}$ and $k$ of its elements, say $e_1,\ldots,e_k$, are such that $\rho(e_1),\ldots,\rho(e_k)$ form a basis of $\mathfrak{h}$ because $\rho$ is surjective. Since $[e_i,e_j]=\lambda_{ij}^ie_i+\lambda_{ij}^je_j$ for some constants $\lambda_{ij}^i,\lambda_{ij}^j$, it follows that $[\rho(e_i),\rho(e_j)]=\lambda_{ij}^i\rho(e_i)+\lambda_{ij}^j\rho(e_j)$. Hence, $\mathfrak{h}$ is quasi-rectifiable.
\end{proof}

The above proposition is indeed a method for finding non quasi-rectifiable Lie algebras in view of the behaviour of its quotients by Lie algebra ideals. 

Let us now provide a powerful approach for classifying general quasi-rectifiable Lie algebras. Recall that the equation determining whether a Lie algebra $\mathfrak{g}$ is quasi-rectifiable can be written as follows
\begin{equation}\label{eq:Charact}
\delta\left(\sum_{i=1}^r\lambda_ie^i\right)\wedge\left(\sum_{m=1}^r\lambda_me^m\right)=-\frac 12\sum_{1\leq i,j,k,m\leq r}\lambda_i\lambda_mc_{jk}^ie^j\wedge e^k\wedge e^m=0,
\end{equation}
where, as is standard, we assume that $\{e_1,\ldots,e_r\}$ is a basis of $\mathfrak{g}$ with structure constants $c_{ij}^k$, one has the dual basis $\{e^1,\ldots,e^r\}$ in $\mathfrak{g}^*$, and we consider $v=\sum_{i=1}^r\lambda_ie^i$ to be one of the elements of the dual basis to a quasi-rectifiable basis of $\mathfrak{g}$.

If we choose a basis of three-vectors for $\mathfrak{g}^*$, the coefficients of (\ref{eq:Charact}) in such a basis must be zero. This means that the coordinates of $v$ are solutions of a series of quadratic polynomial equations in the coordinates of $v$ in the chosen basis $\{e^1,\ldots,e^r\}$ of $\mathfrak{g}^*$. More specifically, (\ref{eq:Charact}) can be rewritten as 
\begin{equation*}
\delta\left(\sum_{i=1}^r\lambda_ie^i\right)\wedge\left(\sum_{m=1}^r\lambda_me^m\right)=-\sum_{1\leq j<k<m\leq r}\sum_{i=1}^r\lambda_i(\lambda_mc_{jk}^i+\lambda_jc_{km}^i+\lambda_kc_{mj}^i)e^j\wedge e^k\wedge e^m=0,
\end{equation*}
which allows us to define some polynomials
$$
P_{jkm}(\lambda_1,\ldots,\lambda_r)=\sum_{i=1}^r\lambda_i(\lambda_mc_{jk}^i+\lambda_jc_{km}^i+\lambda_kc_{mj}^i),\qquad 1\leq j<k<m\leq r,
$$
that must be zero on the elements of a rectifiable basis for $\mathfrak{g}$. The above polynomials can easily be derived through a mathematical computation program, and were used to construct the classification of four- and five-dimensional indecomposable Lie algebras detailed in Tables \ref{fig:Lie-algebra-classification4}, \ref{fig:Lie-algebra-classification5},  and \ref{fig:Lie-algebra-classification6}. It is worth noting that one can explain in detail how one of the Lie algebras in the above-mentioned tables is proved to not be quasi-rectifiable. 

If $\mathfrak{g}$ is quasi-rectifiable, then $\mathfrak{g}^*$ admits a basis, dual to a quasi-rectifiable basis of $\mathfrak{g}$, satisfying
$$
\nu^{(\alpha)}=\sum_{\mu=1}^r\lambda^{(\alpha)}_{\mu} e^\mu,\qquad \alpha=1,\ldots,r,
$$
and
$$
P_{jkm}(\lambda^{(\alpha)}_{1},\ldots,\lambda^{(\alpha)}_{r})=0,\qquad 1\leq j,k,m\leq r,\qquad \alpha=1,\ldots,r.
$$
Since $\nu^{(1)},\ldots,\nu^{(r)}$ form a basis, the $(r\times r)$-matrix of coefficients $\lambda^{(\alpha)}_\mu$ must have a determinant different from zero. Nevertheless, it frequently happens that one of the polynomials $P_{jkm}$ is such that all its zeros have a coordinate equal to zero, e.g. $-\lambda_1^2=0$. Let us prove by contradiction that the associated $\mathfrak{g}$ does not admit a quasi-rectifiable basis. If $\mathfrak{g}$ admitted such a basis $\nu^{(1)},\ldots,\nu^{(r)}$, all its elements would have a coordinate equal to zero. Then, the matrix of their coefficients would be equal to zero, and they could not form a basis. Hence, the Lie algebra $\mathfrak{g}$ is non quasi-rectifiable, which is a contradiction.

\subsection{On quasi-rectifiable two- and three-dimensional Lie algebras}
Let us classify all two- and three-dimensional quasi-rectifiable Lie algebras.

\begin{proposition} Every one- and two-dimensional Lie algebra is quasi-rectifiable and every basis is quasi-rectifiable.
\end{proposition}
\begin{proof}
All one-dimensional Lie algebras are Abelian and therefore quasi-rectifiable. Moreover, every basis is quasi-rectifiable because it has only one element. Two-dimensional Lie algebras have a basis of two elements, let us say $\{e_1,e_2\}$. Then,
$
[e_1,e_2]\in \langle e_1,e_2\rangle,
$
and the basis is quasi-rectifiable. Hence, every two-dimensional Lie algebra is also quasi-rectifiable.
\end{proof}

There are many other Lie algebras that can be shown to be quasi-rectifiable. A direct inspection of the commutation relations of $\mathfrak{s}_{4,3}$ in Table \ref{fig:Lie-algebra-classification4} shows that it is quasi-rectifiable.

If a Lie algebra of vector fields is quasi-rectifiable, then it is, as an abstract Lie algebra, quasi-rectifiable too. Despite that, there may exist a Lie algebra of vector fields that gives rise to a quasi-rectifiable abstract Lie algebra, but it is not a quasi-rectifiable Lie algebra of vector fields because the elements of its basis are not linearly independent at a generic point. For instance, consider the Lie algebra $V_2$ of vector fields on $\mathbb{R}$ spanned by 
$$
X_1=\frac{\partial}{\partial x},\qquad X_2=x\frac{\partial}{\partial x}.
$$
As an abstract Lie algebra, it is a quasi-rectifiable one because it is two-dimensional. On the other hand, any basis of vector fields $Y_1,Y_2$ of $V_2$ satisfies that $Y_1\wedge Y_2=0$. Hence, $V_2$ is not a quasi-rectifiable Lie algebra of vector fields.

%Concerning Lie algebras of vector fields on the plane, the GKO classification shows that there exist two classes of two-dimensional Lie algebras of vector fields on the plane. All of them are quasi-rectifiable.

 The values of $\vartheta \wedge \delta \vartheta$ have been determined for every three-dimensional Lie algebra due to the fact that it characterises certain contact forms \cite{LR22}. In our case, this will serve to establish whether we can obtain quasi-rectifiable three-dimensional Lie algebras. Let us give a first result to characterise quasi-rectifiable Lie algebra structures on three-dimensional Lie algebras. 

\begin{proposition} Let $\delta:\mathfrak{g}^*\rightarrow \mathfrak{g}^*\wedge\mathfrak{g}^*$ be equal to minus the transpose of the Lie bracket on a three-dimensional Lie algebra $\mathfrak{g}$. Then,  $\mathfrak{g}$ is quasi-rectifiable if the zeros of the three-vector function $\Upsilon:\vartheta\in \mathfrak{g}^*\mapsto \vartheta\wedge \delta \vartheta\in  \Lambda^3\mathfrak{g}^*$ are not contained in a plane of $\mathfrak{g}$. 
\end{proposition} 
\begin{proof} By applying  Theorem \ref{Th::RectCri} successively, one obtains that $\Upsilon(\vartheta)$ is a second-order polynomial in the coefficients of $\vartheta$ in a stable basis. The dual to a rectification basis is given by three zeros of  $\Upsilon(\vartheta)$ that must be linearly independent. Hence, they exist if and only if the set of zeros of $\Upsilon (\vartheta)$ is not contained in a plane, namely when they span a subspace of dimension three or higher in $\mathfrak{g}$. 
\end{proof}
$\bullet$ Case $\mathfrak{sl}_2$: 
The corresponding Lie bracket is an antisymmetric bilinear function that can be understood uniquely as a mapping $[\cdot,\cdot]:v\wedge w\in\mathfrak{sl}_2\wedge \mathfrak{sl}_2\mapsto [v,w]\in \mathfrak{sl}_2$. Defining the map $\delta:\mathfrak{sl}_2^*\rightarrow \mathfrak{sl}_2^*\wedge\mathfrak{sl}_2^*$ as $\delta = -[\cdot,\cdot]^T$, one has, in particular, that $\delta(e^1)\in \mathfrak{sl}_2^*\wedge \mathfrak{sl}_2^*$. Take the basis $\{e_1,e_2,e_3\}$ of $\mathfrak{sl}_2$ appearing in Table \ref{fig:Lie-algebra-classification} and consider its dual basis $\{e^1,e^2,e^3\}$. Then,
$$
\begin{gathered}
\delta(e^1)(e_1\wedge e_2)=-e^1([e_1,e_2])=0\Rightarrow \delta(e^1)(e_1\wedge e_2)=0,\\
\delta(e^1)(e_1\wedge e_3)=-e^1([e_1,e_3])=0\Rightarrow \delta(e^1)(e_1\wedge e_3)=0,\\
\delta(e^1)(e_2\wedge e_3)=-e^1([e_2,e_3])=-1\Rightarrow \delta(e^1)(e_2\wedge e_3)=-1.    
\end{gathered}
$$
Since we define
$
\vartheta_1\wedge \vartheta_2(v_1,v_2)=\vartheta_1(v_1)\vartheta_2(v_2)-\vartheta_1(v_2)\vartheta_2(v_1)
$ for every $\vartheta_1,\vartheta_2\in \mathfrak{g}^*$ and $v_1,v_1\in \mathfrak{g}$ and
$$
(\vartheta_1\wedge \vartheta_2)(v_1\wedge v_2)=\det \left[\begin{array}{cc}\vartheta_1(v_1)&\vartheta_1(v_2)\\\vartheta_2(v_1)&\vartheta_2(v_2)\end{array}\right],
$$
we have $\delta(e^1)=e^3\wedge e^2$ since both sides act in the same manner as $\mathbb{R}$-valued bilinear mappings on $\mathfrak{sl}_2$. Proceeding similarly for $\delta(e^2),\delta(e^3)$, we obtain
\begin{gather*}
    \delta(e^1) = -e^1([\cdot,\cdot])=    e^3\wedge e^2\,,\quad
    \delta(e^2) = -e^2([\cdot,\cdot]) = - e^1\wedge e^2\,,\\
    \delta(e^3) = -e^3([\cdot,\cdot]) =  e^1\wedge e^3\,.
\end{gather*}
Using the standard isomorphism ${\rm Hom}(\mathfrak{sl}_2^*,\mathfrak{sl}_2^*\wedge \mathfrak{sl}_2^*)\simeq \mathfrak{sl}_2\otimes (\mathfrak{sl}_2^*\wedge \mathfrak{sl}_2^*)$, one has that 
$$ \delta =  e_1\otimes e^3\wedge e^2 -  e_2\otimes e^1\wedge e^2 +  e_3\otimes e^1\wedge e^3\,. $$
In this case, one can write a general element of $\mathfrak{g}^*$ as $\lambda_1e^1+\lambda_2e^2+\lambda_3e^3$ for some constants $\lambda_1,\lambda_2,\lambda_3$. Then, the equation (\ref{Eq:ConRec}) is
\begin{multline*}
    0 = \delta(\lambda_1e^1 + \lambda_2e^2 + \lambda_3e^3)\wedge(\lambda_1e^1 + \lambda_2e^2 + \lambda_3e^3)\\= \left( \lambda_1 e^3\wedge e^2 - \lambda_2e^1\wedge e^2 +  \lambda_3 e^1\wedge e^3 \right)\wedge(\lambda_1e^1 + \lambda_2e^2 + \lambda_3e^3)
    \\= -\left(\lambda_1^2 + 2\lambda_2\lambda_3\right) e^1\wedge e^2\wedge e^3\,.
\end{multline*}
Then, we look for a basis of elements of $\mathfrak{sl}_2^*$  such that $\lambda_1^2 + 2\lambda_2\lambda_3= 0$. For instance, one can choose the dual elements
$$
\vartheta^1=e^1-\frac 12 e^2+e^3,\qquad \vartheta^2=e^2,\qquad \vartheta^3=e^3
 $$   whose dual basis
 $$
 \vartheta_1=e_1,\qquad \vartheta_2=e_2+\frac 12 e_1,\qquad \vartheta_3=e_3-e_1
 $$
 satisfies  
 $$
 [\vartheta_1,\vartheta_2]=\vartheta_2,\qquad [\vartheta_1,\vartheta_3]=-\vartheta_3-\vartheta_1,\qquad [\vartheta_2,\vartheta_3]=-\frac 12\vartheta_3+\vartheta_2,
 $$
which shows that $\mathfrak{sl}_2$ is quasi-rectifiable. Note moreover that $\delta \vartheta^i\wedge\vartheta^i=0$ for $i=1,2,3$.
 
$\bullet$ Case $\mathfrak{r}_{3,\lambda}$, with  $\lambda\in (-1,1)$. As previously, define the map $\delta:\mathfrak{r}_{3,\lambda}^*\rightarrow \mathfrak{r}_{3,\lambda}^*\wedge\mathfrak{r}_{3,\lambda}^*$ as $\delta = -[\cdot,\cdot]^T$. Then,
\begin{gather*}
    \delta(e^1) = -e^1([\cdot,\cdot])=    e^1\wedge e^3\,,\quad
    \delta(e^2) = -e^2([\cdot,\cdot]) = - e^3\wedge e^2\,,\\
    \delta(e^3) = -e^3([\cdot,\cdot]) = 0\,,
\end{gather*}
and thus,
$$
    \delta =  e_1\otimes e^1\wedge e^3 -  \lambda e_2\otimes e^3\wedge e^2\,.
$$
Therefore,
\begin{multline*}
    0 = \delta(\lambda_1e^1 + \lambda_2\lambda e^2 + \lambda_3e^3)\wedge(\lambda_1e^1 + \lambda_2e^2 + \lambda_3e^3)\\= \left(  {\lambda_1}e^1\wedge e^3 - {\lambda_2\lambda} e^3\wedge e^2\right)\wedge(\lambda_1e^1 + \lambda_2e^2 + \lambda_3e^3)\\
    = \lambda_1\lambda_2(1-\lambda)e^1\wedge e^2\wedge e^3\,.
\end{multline*}
Then, the left-invariant contact forms on a Lie group with Lie algebra isomorphic to $\mathfrak{r}_{3,\lambda}$  are characterised by the condition $\lambda_1\lambda_2\neq 0$. 
In this case, an adequate basis for the dual is

$$
\vartheta^1=e^3,\qquad \vartheta^2=e^1,\quad \vartheta^3=e^2,
$$
which gives a basis for the Lie algebra given by
$$
\vartheta_1=e_1,\quad \vartheta_2=e_2,\quad \vartheta_3=e_3.
$$

$\bullet$ Case $\mathfrak{r}'_{3,\lambda\neq 0}$. Defining the map $\delta:\mathfrak{r}_{3,\lambda\neq 0}^{\prime \,*}\rightarrow \mathfrak{r}_{3,\lambda\neq 0}^{\prime \,*}\wedge\mathfrak{r}_{3,\lambda\neq 0}^{\prime \,*}$ as $\delta = -[\cdot,\cdot]^T$, we have
$$
    \delta(e^1) =  \lambda e^1\wedge e^3 -  e^3\wedge e^2\,,\quad
    \delta(e^2) = - e^1\wedge e^3 -  \lambda e^3\wedge e^2\,,\quad
    \delta(e^3) = 0\,,
$$
and thus,
$$
    \delta =  \lambda e_1\otimes e^1\wedge e^3 -  e_1\otimes e^3\wedge e^2 -  e_2\otimes e^1\wedge e^3 -  \lambda e_2\otimes e^3\wedge e^2\,.
$$
In this case,
\begin{multline*}
    0 = \delta(\lambda_1e^1 + \lambda_2e^2 + \lambda_3e^3)\wedge(\lambda_1e^1 + \lambda_2e^2 + \lambda_3e^3)\\
    = \left(  \lambda\lambda_1  e^1\wedge e^3 -  \lambda_1 e^3\wedge e^2 -  \lambda_2 e^1\wedge e^3 -  \lambda\lambda_2 e^3\wedge e^2 \right)\wedge(\lambda_1e^1 + \lambda_2e^2 + \lambda_3e^3)\\
    =  \left( \lambda_1^2 + \lambda_2^2\right) e^1\wedge e^2\wedge e^3\,.
\end{multline*}
Then,  only one linearly independent covector, $e^3$, satisfies the chosen properties. Hence, no basis can be chosen and $\mathfrak{r}'_{3,\lambda\neq 0}$ is not quasi-rectifiable.

The other cases can be computed similarly, as summarised in the following theorem.

\begin{theorem} The rectification polynomials for non-Abelian three-dimensional Lie algebras and the classification of quasi-rectifiable Lie algebras are given in Table  \ref{fig:Lie-algebra-classification}.
\end{theorem}
\begin{table}[ht]
    \centering
    \begin{tabular}{|c|c|c|c|c|c|}
        \hline
        Lie algebra $\mathfrak{g}$& $[e_1,e_2]$ & $[e_1,e_3]$ & $[e_3,e_2]$ & rectification polynomials& quasi-rectifiable \\
        \hline
        $\mathfrak{sl}_2$ & $e_2$ & $-e_3$ & $-e_1$ & $\lambda_1^2 + 2\lambda_2\lambda_3$ & Yes\\
        \hline
        $\mathfrak{su}_2$ & $e_3$ & $-e_2$ & $-e_1$ & $\lambda_1^2 + \lambda_2^2 + \lambda_3^2$& No \\
        \hline
        $\mathfrak{h}_3$ & $e_3$ & $0$ & $0$ & $\lambda_3$ & No \\
        \hline
        $\mathfrak{r}'_{3,0}$ & $-e_3$ & $e_2$ & $0$ & $\lambda_2^2 + \lambda_3^2 $& No \\
        \hline
        $\mathfrak{r}_{3,-1}$ & $e_2$ & $-e_3$ & $0$ & $\lambda_2\lambda_3$ & Yes\\
        \hline
        $\mathfrak{r}_{3,1}$ & $e_2$ & $e_3$ & $0$ & 0& Yes\\
        \hline
        $\mathfrak{r}_{3}$ & $0$ & $-e_1$ & $e_1 + e_2$ & $\lambda_1 $ & No\\
        \hline
        $\mathfrak{r}_{3,\lambda}$ & $0$ & $-e_1$ & $\lambda e_2$ & $ \lambda_1\lambda_2$ & Yes\\
        \hline
        $\mathfrak{r}'_{3,\lambda\neq 0}$ & $0$ & $e_2-\lambda e_1$ & $\lambda e_2 + e_1$ & $\lambda_1^2 + \lambda_2^2$ & No \\
        \hline
    \end{tabular}
    \caption{Classification of quasi-rectifiable non-Abelian three-dimensional Lie algebras. Note that $\lambda\in (-1,1).$ The value of a polynomial determining the solutions of $0=\vartheta\wedge \delta \vartheta$ for $\vartheta=\sum_{i=1}^3\lambda_ie^i$ for the dual basis $\{e^1,e^2,e^3\}$ to the basis $\{e_1,e_2,e_3\}$ of the Lie algebra $\mathfrak{g}$ is given in Table 1.}
    \label{fig:Lie-algebra-classification}
\end{table}

A first look at Table \ref{fig:Lie-algebra-classification} shows that the fact that a Lie algebra is quasi-rectifiable has nothing to do with whether the Lie algebra is simple or not. In fact, $\mathfrak{sl}_2$ and $\mathfrak{su}_2$ are both simple, but one is quasi-rectifiable while the other is not. The fact that a Lie algebra is quasi-rectifiable has nothing to do with whether a Lie algebra is solvable either. In fact, $\mathfrak{r}_3$ and $\mathfrak{r}_{3,\lambda}$ are solvable, but $\mathfrak{r}_3$ is not quasi-rectifiable while $\mathfrak{r}_{3,\lambda}$ is. 

Note that complexifications of the quasi-rectifiable Lie algebras are also quasi-rectifiable.

\subsection{On four-, five- and higher-dimensional quasi-rectifiable Lie algebras}
Let us provide some results on four- and higher-dimensional quasi-rectifiable indecomposable Lie algebras.   First, let us classify four-dimensional indecomposable quasi-rectifiable Lie algebras. In this case, we will use the results in \cite{PW77}, where all  indecomposable Lie subalgebras up to dimension six are determined. This is very useful, since the knowledge of the internal structure of quasi-rectifiable Lie algebras allows us to easily determine whether many of them are quasi-rectifiable or not.

\begin{proposition}
    The Lie algebras $\mathfrak{n}_{5,1},\mathfrak{n}_{5,2},\mathfrak{n}_{5,4},\mathfrak{n}_{5,5},\mathfrak{n}_{5,6}$ are non quasi-rectifiable.
\end{proposition}
\begin{proof}
    The quotients of the Lie algebras $\mathfrak{n}_{5,2},\mathfrak{n}_{5,5},\mathfrak{n}_{5,6}$ 
by $I=\langle e_1\rangle $ give rise to four-dimensional Lie algebras with a basis $\{\widetilde{e}_2=e_2+I,\widetilde{e}_3=e_3+I,\widetilde{e}_4=e_4+I,\widetilde{e}_5=e_5+I\}$ and nonvanishing commutation relations
$$
\mathfrak{n}_{5,2}/I:\qquad %[\widetilde{e}_2,\widetilde{e}_3]=[\widetilde{e}_2,\widetilde{e}_4]=[\widetilde{e}_2,\widetilde{e}_5]=[\widetilde{e}_3,\widetilde{e}_5]=0,
[\widetilde{e}_3,\widetilde{e}_4]=\widetilde{e}_2,  \quad [-\widetilde{e}_5,\widetilde{e}_4]=\widetilde{e}_3.
$$
$$
\mathfrak{n}_{5,5}/I, \mathfrak{n}_{5,6}/I:\qquad %[\widetilde{e}_2,\widetilde{e}_3]=[\widetilde{e}_2,\widetilde{e}_4]=[\widetilde{e}_2,\widetilde{e}_5]=[\widetilde{e}_3,\widetilde{e}_4]=0,
[\widetilde{e}_3,\widetilde{e}_5]=\widetilde{e}_2, \quad [\widetilde{e}_4,\widetilde{e}_5]=\widetilde{e}_3.
$$
% $$
% \mathfrak{n}_{5,4}/I:\qquad [\widetilde{e}_2,\widetilde{e}_3]=[\widetilde{e}_2,\widetilde{e}_4]=[\widetilde{e}_2,\widetilde{e}_5]=[\widetilde{e}_3,\widetilde{e}_5]=[\widetilde{e}_3,\widetilde{e}_4]=0,  [\widetilde{e}_4,\widetilde{e}_5]=\widetilde{e}_2.
% $$
All of the above Lie algebras are isomorphic to $\mathfrak{n}_{4,1}$, which is non quasi-rectifiable. In fact, the non-vanishing commutation relations of $\mathfrak{n}_{4,1}$ are
$$
[e_2,e_4]=e_1,\qquad [e_3,e_4]=e_2,
$$
and since $e_1$ is an ideal of $\mathfrak{n}_{4,1}$, there exists a Lie algebra surjective projection $\rho:\mathfrak{n}_{4,1}\rightarrow \mathfrak{n}_{4,1}/\langle e_1\rangle\simeq \mathfrak{h}_3$ into a Lie algebra that is not quasi-rectifiable (see Table \ref{fig:Lie-algebra-classification}). Hence, $\mathfrak{n}_{4,1}$ is not quasi-rectifiable. In turn, $\mathfrak{n}_{5,2},\mathfrak{n}_{5,5},\mathfrak{n}_{5,6}$ are not rectifiable.

The quotients of the Lie algebras $\mathfrak{n}_{5,1},\mathfrak{n}_{5,4}$ by $I=\langle e_1,e_3\rangle $ form a Lie algebra isomorphic to $\mathfrak{h}_{3}$, which is non quasi-rectifiable. In fact, the nonvanishing commutation relations are
$$
\mathfrak{n}_{5,1}/I, \mathfrak{n}_{5,4}/I:\quad %\qquad [\widetilde{e}_2,\widetilde{e}_3]=[\widetilde{e}_2,\widetilde{e}_4]=[\widetilde{e}_2,\widetilde{e}_5]=[\widetilde{e}_3,\widetilde{e}_5]=0,[\widetilde{e}_3,\widetilde{e}_4]=0,  
[\widetilde{e}_4,\widetilde{e}_5]=\widetilde{e}_2.
$$\end{proof}

\begin{theorem} The rectification polynomials for non-Abelian four-dimensional Lie algebras and the classification of quasi-rectifiable Lie algebras are given in Table  \ref{fig:Lie-algebra-classification4}.
\end{theorem}

%\begin{proposition} The direct sum of quasi-rectifiable Lie algebras is quasi-rectifiable.
%\end{proposition}

It can be proved that obtaining quasi-rectifiable simple Lie algebras of higher order is an involved task. Let us consider $\mathfrak{sl}_3$. In this case, one can choose the basis
$$
H_1=\frac{1}{2} \left(
\begin{array}{ccc}
 1 & 0 & 0 \\
 0 & -1 & 0 \\
 0 & 0 & 0 \\
\end{array}
\right),\qquad H_2=\left(
\begin{array}{ccc}
 0 & 1 & 0 \\
 0 & 0 & 0 \\
 0 & 0 & 0 \\
\end{array}
\right),\qquad H_3=\left(
\begin{array}{ccc}
 0 & 0 & 0 \\
 -1 & 0 & 0 \\
 0 & 0 & 0 \\
\end{array}
\right),$$
$$
H_4=\frac1{2 \sqrt{3}}{\left(
\begin{array}{ccc}
 1 & 0 & 0 \\
 0 & 1 & 0 \\
 0 & 0 & -2 \\
\end{array}
\right)},\qquad H_5=\left(
\begin{array}{ccc}
 0 & 0 & 0 \\
 0 & 0 & 1 \\
 0 & 0 & 0 \\
\end{array}
\right),\qquad H_6=\left(
\begin{array}{ccc}
 0 & 0 & 0 \\
 0 & 0 & 0 \\
 0 & -1 & 0 \\
\end{array}
\right),
$$
$$
H_7=\left(
\begin{array}{ccc}
 0 & 0 & 1 \\
 0 & 0 & 0 \\
 0 & 0 & 0 \\
\end{array}
\right),\qquad H_8=\left(
\begin{array}{ccc}
 0 & 0 & 0 \\
 0 & 0 & 0 \\
 -1 & 0 & 0 \\
\end{array}
\right).
$$
Note that  $H_1$ and $H_4$ are such that the remaining elements of the basis are eigenvectors of $H_1,H_4$. Their eigenvalues allow us to put such elements in the edges of Figure  \ref{Fig:Hexagon}. In fact, the eigenvalue relative to $H_1$ gives the coordinate in the $X$ axis, while the eigenvalue with respect to $H_4$ sets the coordinate in the $Y$ axis. Additionally, 
$$
\begin{gathered}
[H_1,H_2]=H_1,\qquad [H_2,H_4]=0,\qquad [H_1,H_3]=-H_3,\qquad [H_3,H_4]=0,\qquad \\ [H_1,H_5]=-H_5/2,\quad [H_4,H_5]=\sqrt{3}/2H_5,\quad [H_1,H_6]=H_6/2,\quad [H_4,H_6]=-\sqrt{3}/2H_6,\\
[H_1,H_7]=H_7/2,\qquad [H_4,H_7]=\sqrt{3}/2H_5,\quad [H_1,H_8]=-H_8/2,\quad [H_4,H_8]=-\sqrt{3}/2H_6.
\end{gathered}
$$
%$$
%[H_5,H_2]=H_7,\qquad [H_3,H_7]=H_5,\qquad [H_8,H_2]=H_6.
%$$

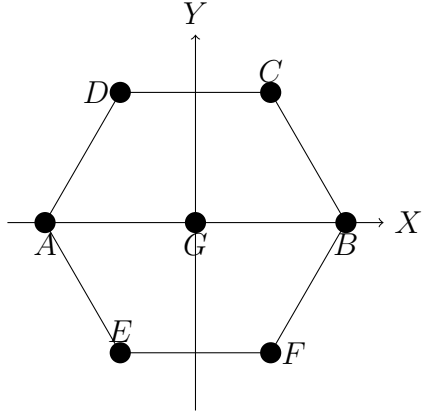
\begin{figure}
  \centering
  \begin{tikzpicture}[scale=2]
    % Draw X and Y axes
    \draw[->] (-1.25,0) -- (1.25,0) node[right] {\(X\)};
    \draw[->] (0,-1.25) -- (0,1.25) node[above] {\(Y\)};
    
    % Coordinates of the hexagon vertices
    \coordinate (B) at (1,0);
    \coordinate (C) at ({1*cos(60)}, {1*sin(60)});
    \coordinate (D) at ({1*cos(120)}, {1*sin(120)});
    \coordinate (A) at (-1,0);
    \coordinate (E) at ({1*cos(240)}, {1*sin(240)});
    \coordinate (F) at ({1*cos(300)}, {1*sin(300)});
    \coordinate (G) at (0,0); % Center point
    
    % Draw the hexagon
    \draw (B) -- (C) -- (D) -- (A) -- (E) -- (F) -- cycle;
    
    % Vertex labels
    \foreach \vertex/\position in {A/below,B/below,C/above,D/left,E/above,F/right,G/below}
        \fill (\vertex) circle (2pt) node[\position] {\(\vertex\)};
  \end{tikzpicture}
  \caption{Hexagon centred in the \(\mathbb{R}^2\) plane with three vertices on the horizontal axis. This represents the root diagram for $\mathfrak{sl}_3$. Note that $H_1,H_4$ belong to $G$, while $B$ and $A$ contain $H_2$ and $H_3$, respectively. Moreover, $C$, $D$, $E$ and $F$ have $H_7$, $H_5$, $H_6$ and $H_8$, respectively. }\label{Fig:Hexagon}
\end{figure}
Recall that the condition for obtaining a quasi-rectifiable Lie algebra is to obtain a set of eight linearly independent elements of $\mathfrak{sl}_3^*$ satisfying the equation 
$$
\vartheta\wedge \delta \vartheta=0.
$$
However, the above equation can be written for certain  elements of $\mathfrak{g}$. For instance,
\begin{equation}\label{Eq:Rel}
(\vartheta \wedge \delta \vartheta) (v_1,v_2,v_3)=-\sum_{\sigma \in S^3}\epsilon_{\sigma}\vartheta(v_{\sigma(1)})\vartheta([v_{\sigma(2)},v_{\sigma(3)}]),
\end{equation}
where $S^3$ is the space of three-element permutations, $\epsilon_\sigma$ stands for the sign of the permutation $\sigma$, and  $v_1,v_2,v_3\in \mathfrak{sl}_3$. Consider the dual basis to the basis of $\mathfrak{sl}_3$. Consider that all the elements along the border of the $\mathfrak{sl}_3$ diagram can always be obtained as the Lie bracket of any two other elements. By taking (\ref{Eq:Rel}) for that, one obtains $\vartheta=\sum_{\alpha=1}^8\lambda_\alpha\vartheta^\alpha$, and 
$$
\vartheta\wedge \delta\vartheta(v_5,v_7,v_2)=-\lambda_7^2=0.
$$
The same is true for any other element $\lambda_\alpha$ over the boundary of the polytope.
 Then, $\vartheta=\lambda_1\vartheta_1+\lambda_4\vartheta_4$. This is insufficient to obtain eight elements in $\mathfrak{sl}_3$, which is not quasi-rectifiable. Note that the same approach can be applied to many other Lie algebras. In particular, this can be proved to be true for $\mathfrak{sl}_4$, whose study can be reduced to $\mathfrak{sl}_3$ (see Figure \ref{Fi:sl4}), $\mathfrak{so}(2n,\mathbb{C})$, $\mathfrak{so}(2n,\mathbb{R})$ for $n\geq 2$, etcetera.

\begin{figure}
  \centering
 \includegraphics[scale=0.9]{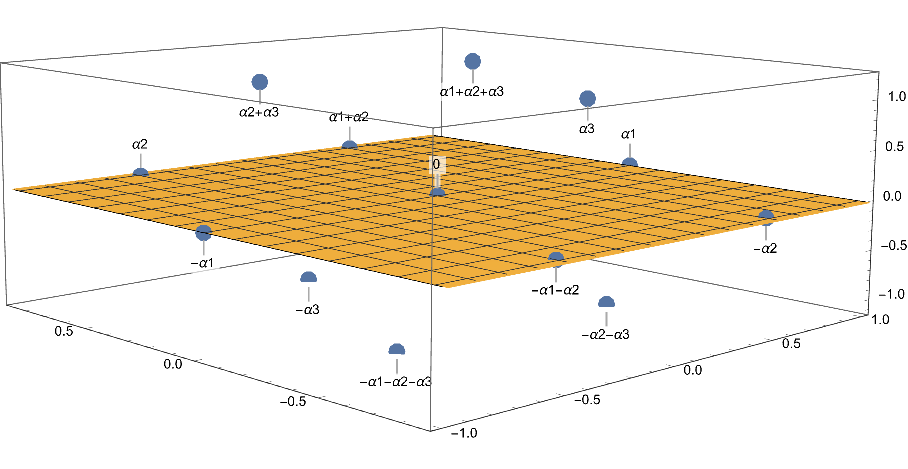} 
  \caption{Schematic root diagram for $\mathfrak{sl}_4$, which has 6 positive roots $\alpha_1,\alpha_2,\alpha_3,\alpha_1+\alpha_2,\alpha_2+\alpha_3,\alpha_1+\alpha_2+\alpha_3$ (cf. \cite{GHLMMRRT20}). Each root is represented by its eigenvalues relative to a basis of the Cartan algebra for $\mathfrak{sl}_4$. Note that the plane $z=0$ represents the root diagram for $\mathfrak{sl}_3$.} \label{Fi:sl4}
\end{figure}

\section{Applications to hydrodynamic-type systems}\label{Se:Hydro}

This section illustrates through examples the connection between the description of the $k$-wave solutions of hydrodynamic-type equations, quasi-rectifiable families of vector fields, quasi-rectifiable Lie algebras of vector fields, and quasi-rectifiable Lie algebras following the theory developed in Sections \ref{Sec:Elastic}, \ref{Sec:Direct}, and \ref{Sec::ClaRecLieAlg}.

Let us discuss this link for homogeneous hyperbolic quasilinear first-order systems of PDEs of the form 
\begin{equation}\label{eq:Hydr}
\sum_{\alpha=1}^q\sum_{i=1}^w(A^l)^{i}_\alpha(u) \frac{\partial u^\alpha}{\partial x^i}=0,\qquad l=1,\ldots,m,
\end{equation}
in $w$ independent variables $x=(x^1,\ldots,x^w)\in\mathbb{R}^w$ and $q$ dependent variables $u=(u^1,\ldots,u^q)\in \mathbb{R}^q$, where $A^{1},
\ldots,A^m$ are $q\times p$ matrix functions depending on $u\in \mathbb{R}^q$. 

In the context of  the $k$-wave solutions of hydrodynamic-type equations, one finds that they are described via quasi-rectifiable families of vector fields that are frequently put in quasi-rectifiable form to obtain solutions \cite{Gr23,GL23}. 

 The $k$-wave solutions of \eqref{eq:Hydr}, obtained via the generalised method of characteristics (GMC), are obtained from the algebraic system
\begin{equation}\label{Eq:Rect}
	\sum_{\alpha=1}^q\sum_{i=1}^wA^{li}_\alpha (u)\lambda^{(s)}_i(x,u)X^\alpha_{(s)}(x,u)=0,\qquad s=1,\ldots,k,\qquad l=1,\ldots,m.
\end{equation}
The wave covectors 
where $\lambda^{(1)},\ldots,\lambda^{(k)}$ are $\mathbb{R}^q$-parametrised differential one-forms on $\mathbb{R}^w$ such that $\lambda^{(1)}\wedge\ldots\wedge\lambda^{(k)}$ does not vanish. Let us assume that, for each fixed $\lambda^{(s)}$, there exists one $\mathbb{R}^w$-parametrised vector field $X_{(s)}$ on $\mathbb{R}^q$.  %First, we solve (\ref{Eq:Rect}) for certain values of $\lambda^{(1)},\ldots,\lambda^{(k)}$ and next we determine the vector fields $X_{(1)},\ldots,X_{(k)}$ which are associated with each of them, respectively. 
In order to obtain the $k$-wave solutions via the GMC, the family of vector fields $X_{(1)},\ldots,X_{(k)}$ has to be quasi-rectifiable. In practical applications, it is assumed that the elements of each pair of vector fields $X_{(1)},\ldots,X_{(k)}$ commute  between themselves. Hence, we rescale these vector fields to ensure that each pair of these vector fields commutes as this is useful, but not necessary, for solving a parametrisation of the solutions of (\ref{eq:Hydr}). Note that such rescaled vector fields do not change the fact that they are solutions of (\ref{Eq:Rect}). To obtain the proper rescaling, one may apply the methods of Section \ref{Sec:Elastic}. It is worth recalling that in order to obtain $k$-waves solutions of (\ref{eq:Hydr}), one may require
$$
(\mathcal{L}_{X_{(\alpha)}}\lambda^{{(\beta)}})_u\in \langle \lambda^{(\beta)}(u),\lambda^{(\alpha)}(u)\rangle,\qquad 1\leq \alpha\neq \beta\leq k.
$$

It is worth noting that the existence of a quasi-rectifiable Lie algebra of vector fields allows for the determination of a parametrisation of the submanifolds related to the $\gamma_{(1)},\ldots,\gamma_{(k)}$ vector fields as follows (cf. \cite[Section 11]{GT96} or \cite[eq. (2.3)]{GZ83b}):

\begin{equation}\label{eq:ExSubm}
\frac{\partial u}{\partial r_\alpha}=\sum_{\beta=1}^kf_{\alpha\beta}(r_1,\ldots,r_k)\gamma_{(\beta)}(u),\qquad \alpha=1,\ldots,k,
\end{equation}
for certain functions $f_{\alpha\beta}(r_1,\ldots,r_k)$. Recall that the above system of PDEs is such that the tangent space to a solution must be the one spanned by the vector fields $\gamma_{(1)},\ldots,\gamma_{(k)}.$ Since such vector fields are linearly independent, one has to additionally assume that
\begin{equation}\label{eq:CondTangeSol}
\sum_{\beta=1}^k(f_{1\beta}(r_1,\ldots,r_k)\gamma_{(\beta)}(u))\wedge \cdots\wedge\sum_{\beta=1}^k( f_{k\beta}(r_1,\ldots,r_k)\gamma_{(\beta)}(u))\neq 0.
\end{equation}
Note that, as a consequence of Theorem \ref{Th::RecBasis}, each pair of different vector fields $X_{(i)},X_{(j)}$ gives rise to a double wave solution, which produces $r(r-1)/2$ different double waves passing through each point $x\in N$. Hence, the construction of double waves is part of the construction of $k$-wave solutions. 

The classical approach providing an integrable system (\ref{eq:ExSubm}) consists of rectifying the vector fields $\gamma_{(1)},\ldots,\gamma_{(k)}$ by multiplying them by non-vanishing functions $f_1(u),\ldots,f_k(u)$ depending on the dependent variables so as to obtain a system of PDEs of the form
$$
\frac{\partial u}{\partial r_\alpha}=f_\alpha(u)\gamma_{(\alpha)}(u),\qquad \alpha=1,\ldots,k,
$$
which is seen to be integrable. 

Nevertheless, one may consider that the rectification of the vector fields $\gamma_{(1)},\ldots,\gamma_{(k)}$ in (\ref{eq:ExSubm}) can be too involved and one may, instead, find some integrable expressions (\ref{eq:ExSubm}) in another way. In this case, one has to choose the coefficient functions $f_{\alpha\beta}(r_1,\ldots,r_k)$ with $\alpha,\beta=1,\ldots,k$ so that (\ref{eq:ExSubm}) is integrable and (\ref{eq:CondTangeSol}) is satisfied.

Let us illustrate the above fact. First, the system (\ref{eq:ExSubm}) is integrable if and only if 
\begin{equation}\label{eq:IntCon}
\left[\frac{\partial}{\partial r_\alpha}+\sum_{\beta=1}^kf_{\alpha\beta}\gamma_{(\beta)},\frac{\partial}{\partial r_\mu}+\sum_{\nu=1}^kf_{\mu\nu}\gamma_{(\nu)}\right]=0,\qquad 1\leq \alpha<\mu\leq k.
\end{equation}
Hence, 
$$
0=\sum_{\beta=1}^k\left(\frac{\partial f_{\alpha\beta}}{\partial r_\mu}-\frac{\partial f_{\mu\beta}}{\partial r_\alpha}+\sum_{\sigma,\delta=1}^kf_{\alpha \sigma}f_{\nu \delta}c_{\sigma\delta}^\beta \right)\gamma_{(\beta)},\qquad 1\leq \alpha<\mu\leq k,
$$
where $[\gamma_{(\alpha)},\gamma_{(\beta)}]=\sum_{\gamma=1}^kc_{\alpha\beta}^\delta\gamma_{(\delta)}$ for certain constants $c_{\alpha\beta}^\delta$ for $\alpha,\beta,\delta=1,\ldots,k$. Note that in hydrodynamic-type systems, we have $\gamma_{(1)}\wedge\ldots\wedge\gamma_{(k)}\neq 0$. 
Thus, 
$$
0=\frac{\partial f_{\alpha\beta}}{\partial r_\mu} -\frac{\partial f_{\mu\beta}}{\partial r_\alpha} +\sum_{\sigma,\delta=1}^kf_{\alpha \sigma}f_{\nu \delta}c_{\sigma\delta}^\beta,\qquad 1\leq \alpha<\mu\leq k,\qquad \beta=1,\ldots,k.
$$

Indeed, it has not previously been stated in the literature that (\ref{eq:ExSubm}), if integrable, is a so-called {\it PDE Lie system} \cite{CGM07,CL11}. In other words, it is an integrable  first-order system of PDEs in normal form such that the right-hand side is given by a linear combination of vector fields whose functions depend on the independent variables  spanning a finite-dimensional Lie algebra of vector fields. In our case,  due to the nature of hydrodynamic-type systems and the generalised method of characteristics, this Lie algebra of vector fields is quasi-rectifiable. Moreover, the standard theory of PDE Lie systems \cite{Ra11} can be applied to the study of its properties and solutions.

These and other topics will be illustrated in physical and mathematical examples analysed in the following three subsections.

\subsection{Solutions of (1+1)-dimensional hydrodynamic system}
Let us focus on the hydrodynamic equations in $(1+1)$-dimensions given by a $3\times 3$ function matrix $A(v)$, where $v\in \mathbb{R}_+^3$ for $\mathbb{R}_+=\{x\in \mathbb{R}:x> 0\}$, and there are two independent variables, $x,t$, and three dependent variables, namely $\rho,p,u$, of the form
\begin{equation}\label{eq:11Hydro}
v_t-A(v)v_x=0,\qquad A(v)=\left[\begin{array}{ccc}u&0&\rho\\0&u&\kappa p\\0&1/\rho &u
\end{array}\right].
\end{equation}
Note that  $\{\rho,p,u\}$ are defined on $\mathbb{R}_+^3$. Physically, $\rho$ is the density of the fluid, $p$ is its pressure,  $u$ is the fluid velocity, and $\kappa>0$ is the constant adiabatic exponent.  Some values of the covectors $\lambda$  and the corresponding tangent vectors $\gamma$ for (\ref{eq:11Hydro}) may be given by the pairs
$$
\lambda_+ =\left(u+\sqrt{\frac{\kappa p}{\rho}}\right)dt+dx,\qquad \gamma_+=\rho\frac{\partial}{\partial \rho}+\kappa p\frac{\partial}{\partial p}+\sqrt{\frac{\kappa p}{\rho}}\frac{\partial}{\partial u} ,
$$
$$
\lambda_0 =udt+dx,\qquad \gamma_0=\frac{\partial}{\partial \rho}, 
$$
$$
\lambda_- =\left(u-\sqrt{\frac{\kappa p}{\rho}}\right)dt+dx,\qquad \gamma_-=\rho\frac{\partial}{\partial \rho}+\kappa p\frac{\partial}{\partial p}-\sqrt{\frac{\kappa p}{\rho}}\frac{\partial}{\partial u} .
$$

At any point in $\mathbb{R}_+^3$, one has that $\gamma_+\wedge\gamma_0\wedge\gamma_-\neq 0$. Moreover,
$$
[\gamma_+,\gamma_-]=\frac{1-\kappa}{2}\gamma_++\frac{-1+\kappa}{2}\gamma_-\,,\quad [\gamma_+,\gamma_0]=\frac{1}{4\rho}\gamma_+-\frac{1}{4\rho}\gamma_--\gamma_0,$$$$
[\gamma_0,\gamma_-]=\frac{1}{4\rho}\gamma_+-\frac{1}{4\rho}\gamma_-+\gamma_0.
$$
Consequently, one has that  $\gamma_+,\gamma_-$ is a quasi-rectifiable family of vector fields, while $\gamma_+,\gamma_-,\gamma_0$ is not a quasi-rectifiable family according to the Definition 2.2.  It is worth noting that $\langle\gamma_+,\gamma_-\rangle$ is indeed a quasi-rectifiable Lie algebra.  The vector fields $\gamma_+,\gamma_-$ are associated with  right and left sound waves, and consequently the Lie algebra $\langle \gamma_+,\gamma_-\rangle$ can be called a {\it sound Lie algebra}. As $\gamma_+,\gamma_-$ is a simple family of vector fields, one can put it into a quasi-rectifiable form simply by considering Theorem \ref{Th::RecBasis}. Note that  $\xi=p/\rho^\kappa$ is a constant of motion for $\gamma_+$ and $\gamma_-$. Next, consider a constant of motion for $\gamma_+$ that is not of $\gamma_-$. To obtain it, we write 
$$
\gamma_+=\rho\frac{\partial}{\partial \rho}+\kappa p\frac{\partial}{\partial p}+\sqrt{\kappa 
\xi\rho^{\kappa-1} }\frac{\partial}{\partial u},\qquad
\gamma_-=\rho\frac{\partial}{\partial \rho}+\kappa p\frac{\partial}{\partial p}-\sqrt{\kappa 
\xi\rho^{\kappa-1} }\frac{\partial}{\partial u}
$$
and then, by the method of characteristics and using the fact that $\xi$ is constant along them, one obtains two constants of motion
$$
I_+=\frac{2\sqrt{\kappa\xi \rho^{-1+\kappa}}}{\kappa-1}-u,\qquad I_-=\frac{2\sqrt{\kappa\xi \rho^{-1+\kappa}}}{\kappa-1}+u,
$$
respectively. This allows us to write $\gamma_+$ and $\gamma_-$ in an almost rectified form in the coordinate system $I_+,I_-,\xi$ as 
$$
\gamma_+=2\sqrt{\kappa \xi\rho^{\kappa-1}}\frac{\partial}{\partial I_-},\qquad \gamma_-=2\sqrt{\kappa \xi\rho^{\kappa-1}}\frac{\partial}{\partial I_+}.
$$Then, some functions $h_+$ and $h_-$ can be used to rescale $\gamma_+,\gamma_-$ , respectively, and obtain two commuting vector fields. In particular, we can choose
$$
h_\pm=\frac{1}{2\sqrt{\kappa \xi\rho^{\kappa-1}}}.
$$
It is worth stressing that this rescaling is used in the literature to simplify the parametrisations of surfaces in terms of the Riemann invariants \cite{Gr23}.

Let us use our second method to put the quasi-rectifiable Lie algebra $\langle\gamma_+,\gamma_-\rangle$  into quasi-rectifiable form. In particular, let us apply the Corollary \ref{Cor:Double} and, in this respect, consider the differential one-forms dual to $\gamma_+,\gamma_-$ given by
$$
\eta_+=\frac{dp}{2 \kappa p} +\sqrt{\frac{\rho}{4 \kappa p}}du,\qquad \eta_-=\frac{dp}{2 \kappa p}-\sqrt{\frac{\rho}{4 \kappa p}}du. %\qquad \eta_0=d\rho-\frac{\rho}{\kappa p}dp.
$$
Let us multiply $\eta_+$ by a function $f_+$ so that $f_+ \eta_+$ is the same as the differential of a function on $\mathcal{D}$. 
If $\mathcal{D}$ is the distribution spanned by $\gamma_+,\gamma_-$, one has that 
$$
dI_-=\frac{1}{1-\kappa}\sqrt{\frac{\kappa p}{\rho^3}}d\rho+\frac{1}{\kappa-1}\sqrt{\frac{\kappa}{\rho p}}dp+du,\quad \Xi_-={2\sqrt{\kappa \xi\rho^{\kappa-1}}}\left(\frac{dp}{2 \kappa p} +\sqrt{\frac{\rho}{4 \kappa p}}du\right).
$$
are such that $\Xi_--dI_-$ vanishes on $\gamma_+$ and $\gamma_-$ . In other words $\Upsilon_-|_{\mathcal{D}}=dI_-|_{\mathcal{D}}$. The same applies to 
$$
dI_+=\frac{1}{1-\kappa}\sqrt{\frac{\kappa p}{\rho^3}}d\rho+\frac{1}{\kappa-1}\sqrt{\frac{\kappa}{\rho p}}dp-du,\quad \Xi_+={2\sqrt{\kappa \xi\rho^{\kappa-1}}}\left(\frac{dp}{2 \kappa p} -\sqrt{\frac{\rho}{4 \kappa p}}du\right).
$$
and $\Xi_+-dI_+$ vanishes on $\gamma_+$ and $\gamma_-$. Hence, one obtains that $\gamma_+/[2\sqrt{\kappa \xi\rho^{\kappa-1}}]$ and $\gamma_-/[2\sqrt{\kappa \xi\rho^{\kappa-1}}]$ commute and $I_-$ and $I_+$ put $\gamma_+,\gamma_-$ in quasi-rectifiable form., respectively 

Let us explain how Theorem \ref{Th:FroRecBett} and Corollary \ref{Cor:Double} can be used in a more practical and clarifying manner. The key is that the relation $d(f_i\eta_i)|_\mathcal{D}=dx^i|_\mathcal{D}$ means that the restriction of $f_i\eta_i$ of a leaf of the distribution $\mathcal{D}$ is exact and $x^i$ is a potential on that leaf, but $d(f_i\eta_i)$ does not need to be closed. 

In our practical example, let us write $\eta_+$ in terms of $u,p$ and $\xi$. This shows the form of $\eta_+$ on a leaf of the distribution $\mathcal{D}$, where $u,p$ are coordinates for a constant value of $\xi$. If we multiply $\eta_+$ by a function $f(u,p,\xi)$ so that its  restriction to a leaf of $\mathcal{D}$ is closed, then the potentials depending on $\xi$  constitute a solution of (\ref{Eq:MethodSol}). More specifically, $\eta_+$ take the form 
$$\eta_+=\frac{dp}{2\kappa p} +\sqrt{\frac{\rho}{4p\kappa}}du=\frac{dp}{2\kappa p} +\sqrt{\frac{(p/\xi)^{1/\kappa}}{4 \kappa p}}du
$$
 in the variables $\{p,u,\xi\}$. To solve the equation $f_+\eta_+|_{\mathcal{D}}=dI_+|_{\mathcal{D}}$ it is enough to consider $\xi$ as a constant and to multiply it by $f_+$ so as to obtain the differential of a function that is assumed to depend on the constant $\xi$, i.e.
\begin{multline*}
\sqrt{\frac{4 p\kappa}{(p/\xi)^{1/\kappa}}}\eta_+|_{\mathcal{D}}=du+\sqrt{\frac{\xi^{1/\kappa}}{p\kappa p^{1/\kappa}}}dp|_{\mathcal{D}}\\=du+\sqrt{\frac{\xi^{1/\kappa}}{\kappa}}d\left[\frac{p^{-1/2\kappa+1/2}}{-1/2\kappa+1/2}\right]|_{\mathcal{D}}=d\left(u+\frac{2}{1-\kappa}\sqrt{\frac{p}{\rho\kappa}}\right)|_{\mathcal{D}}=dI_-|_{\mathcal{D}}.    
\end{multline*}
Meanwhile,
\begin{multline*}
\sqrt{\frac{4 p\kappa}{(p/\xi)^{1/\kappa}}}\eta_-|_{\mathcal{D}}=-du+\sqrt{\frac{\xi^{1/\kappa}}{p\kappa p^{1/\kappa}}}dp|_{\mathcal{D}}\\=-du+\sqrt{\frac{\xi^{1/\kappa}}{\kappa}}d\left[\frac{p^{-1/2\kappa+1/2}}{-1/2\kappa+1/2}\right]|_{\mathcal{D}}=d\left(-u+\frac{2}{1-\kappa}\sqrt{\frac{p}{\rho\kappa}}\right)|_{\mathcal{D}}=dI_+|_{\mathcal{D}}.    
\end{multline*}

\subsection{Barotropic fluid flow in $(2k+1)$-dimensions}
Let us study a barotropic fluid flow \cite{Gr23,Mi58}. In this case, we focus on the systems of PDEs on $\mathbb{R}^{1+2k}$ with independent variables $(t,x^1,\ldots,x^{2k})$ and dependent variables $(\rho,u^1,\ldots,u^{2k})$ given by
$$
u_t+(u\cdot\nabla)u=0,\qquad \rho_t+\rho(\nabla u)+(u\cdot \nabla) \rho=0,
$$
where $u\cdot \nabla$ denotes the partial derivative at $x\in \mathbb{R}^{1+2k}$ in the direction given by $u(x)\in \mathbb{R}^{1+2k}$, while $\nabla u$ stands for the standard divergence on $\mathbb{R}^{1+2k}$ of the vector field $\sum_{\alpha=1}^{2k}u^\alpha(x)\partial/\partial x^\alpha+u^0\partial/\partial t$.
In this case, the $\gamma$'s are of the form
$$
\gamma=\gamma_\rho\frac{\partial}{\partial \rho}+\sum_{i=1}^{2k}\gamma^i\frac{\partial}{\partial u^i},
$$
for certain functions $\gamma_\rho,\gamma^1,\ldots,\gamma^{2k}\in C^\infty(\mathbb{R}^{2k+1})$ defined on the space of dependent variables,
while
$$
\lambda=\lambda_0dt+\sum_{i=1}^{2k}\lambda_idx^i
$$
is chosen so that $\lambda_0,\lambda_i$ are functions depending on the dependent variables. 
The conditions ensuring that $\gamma$ and $\lambda$ give rise to a one-wave solution are
$$
\lambda_0=-\sum_{i=1}^{2k}u^i\lambda_i,\qquad \sum_{i=1}^{2k}\gamma^i\lambda_i=0.
$$
One can propose a $k$-wave solution on $\mathbb{R}^{2k+1}$ of the form
$$
\gamma_{(1)}=f_1(u^1,\ldots,u^{2k},\rho) \left(u^1\frac{\partial}{\partial u^1}+u^2\frac{\partial}{\partial u^2}+\rho F_1(u^1,u^2)\frac{\partial}{\partial \rho}\right),
$$
$$
\lambda^{(1)}=\frac{u^2}{\rho}dx^1-\frac{u^1}{\rho}dx^2,
$$
$$
\gamma_{(2)}=f_2(u^1,\ldots,u^{2k},\rho) \left(u^3\frac{\partial}{\partial u^3}+u^4\frac{\partial}{\partial u^4}+\rho F_2(u^3,u^4)\frac{\partial}{\partial \rho}\right),
$$
$$
\lambda^{(2)}=\frac{u^4}{\rho}dx^3-\frac{u^3}{\rho}dx^4,
$$
$$
\ldots\qquad \ldots
$$
$$
\gamma_{(k)}=f_k(u^1,\ldots,u^{2k},\rho) \left(u^{2k-1}\frac{\partial}{\partial u^{2k-1}}+u^{2k}\frac{\partial}{\partial u^{2k}}+\rho F_k(u^{2k-1},u^{2k})\frac{\partial}{\partial \rho}\right),$$
$$
\lambda^{(k)}=\frac{u^{2k}}{\rho}dx^{2k-1}-\frac{u^{2k-1}}{\rho}dx^{2k},
$$
where $F_1,\ldots,F_k$ are arbitrary functions depending on their arguments, while $f_1,\ldots,f_k\in C^\infty(\mathbb{R}^{2k+1})$ are arbitrary functions depending on the dependent variables. One can see that $\gamma_{(1)}\wedge \ldots\wedge \gamma_{(k)}$ and $\lambda^{(1)}\wedge \ldots\wedge \lambda^{(k)}$ are different from zero almost everywhere. Moreover,
$$
[\gamma_{(i)},\gamma_{(j)}]=(\gamma_{(i)}f_j)\gamma_{(j)}-(\gamma_{(j)}f_i)\gamma_{(i)},\qquad 1\leq i<j\leq k.
$$
and 
$$
\sum_{l=1}^{2k+1}\gamma_{(s')}^l\frac{\partial \lambda^{(s)}_j}{\partial u^l}=-F_{s'}(u^{2s'-1},u^{2s'})\lambda^{(s)}_j,\qquad j=1,\ldots,2k+1,\quad 1\leq s'\neq s\leq k,
$$
where $u^{2k+1}=\rho$ and $\lambda_{2k+1}=\lambda_0$, hold. This gives a $k$-wave solution for the barotropic model in $\mathbb{R}^{2k+1}$. 

If we additionally assume that the functions $f_1,\ldots,f_k$ are homogeneous for  each pair of functions $u^{2i-1},u^{2i}$ for $i=1,\ldots,k$, one obtains that $\gamma_{(1)},\ldots,\gamma_{(k)}$ span a quasi-rectifiable Lie algebra of vector fields.

\subsection{$k$-wave solutions  involving quasi-rectifiable Lie algebras}
Let us describe a series of systems of PDEs admitting families of $k$-wave solutions related to quasi-rectifiable Lie algebras of vector fields and constructed via the abstract quasi-rectifiable Lie algebras described in Section \ref{Sec::ClaRecLieAlg}.

Assume that the space of independent variables is $\mathbb{R}^w$ with $w
\geq k$. Consider any of the quasi-rectifiable Lie algebras developed in Section \ref{Sec::ClaRecLieAlg}. Ado's theorem allows one to describe any finite-dimensional Lie algebra as isomorphic to a matrix Lie algebra given by a subspace of $n\times n$ square matrices. Note that $n$ is chosen to be big enough to admit such a representation and it does not need to be equal to the dimension of the Lie algebra to be represented. Let $\{M_1,\ldots,M_k\}$ be a basis of such a matrix Lie algebra. Consider the vector space $\mathbb{R}^{n}$ and the linear coordinates $\{u^1,\ldots,u^n\}$ on it. Define the vector fields
$$
\gamma_{(\alpha)}=-\sum_{\beta,\gamma=1}^n(M_\alpha)^\beta_\gamma u^\gamma \frac{\partial}{\partial u^\beta},\qquad \alpha=1,\ldots,k.
$$
It is known that $\gamma_{(1)},\ldots,\gamma_{(k)}$ span a Lie algebra isomorphic to the one spanned by  $M_1,\ldots,M_k$. In fact, the structure constants of $\gamma_{(1)},\ldots,\gamma_{(k)}$ are the  same as the ones of $M_1,\ldots,M_k$. Consider now the distribution on $\mathbb{R}^n$ spanned by the vector fields $\gamma_{(1)},\ldots,\gamma_{(k)}$ and the annihilator of such a distribution, namely the family of subspaces
$$
\mathcal{A}_U:\{\vartheta\in T^*\mathbb{R}^{n}:\vartheta((\gamma_{(1)})_U)=\ldots=\vartheta((\gamma_{(k)})_U)=0\}\subset T^*\mathbb{R}^n.
$$
It is worth noting that if $\gamma_{(1)}\wedge\ldots\wedge\gamma_{(k)}$ is different from zero at a generic point, it is always possible to make a representation of the initial Lie algebra into a bigger space so that $\gamma_{(1)},\ldots,\gamma_{(k)}$ will be linearly independent at a generic point (cf. \cite{CL11}). Indeed, $n$ can always be chosen to be big enough to ensure that $(\mathcal{A}_U)_x$ is not zero at a generic point. Hence, one defines  the system of PDEs of the form 
$$
\sum_{\alpha=1}^n\sum_{i=1}^w(A^l)^i_\alpha(u)\frac{\partial u^\alpha}{\partial x^i}=0,\qquad l=1,\ldots,m.
$$
Note that  $((A^l)^i_1,\ldots,(A^l)^i_n)$, with $l=1,\ldots,m,$ are chosen so that they will be elements of $\mathcal{A}_U$ and one of them is different from zero. If these conditions are satisfied, it follows that 
$$
\sum_{\alpha=1}^n(A^l)^i_\alpha(u)\gamma^\alpha_{(s)}=0,\quad i=1,\ldots,w,\,\,\,\,\Longrightarrow \,\,\sum_{i=1}^w\sum_{\alpha=1}^n(A^l)^i_{\alpha}(u)\gamma^\alpha_{(s)}\lambda^{(s)}_{i}=0,
$$
for $l=1,\ldots,m$ and $s=1,\ldots,k$. Since the above holds independently of the value of $p$ and the exact value of the coefficients of $\lambda_{(1)},\ldots,\lambda_{(k)}$, one can choose  the coefficients of $\lambda_{(1)},\ldots,\lambda_{(k)}$ so that $\lambda_{(1)}\wedge \ldots\wedge \lambda_{(k)}$ does not vanish. Moreover, one can require the coefficients of the $\lambda_{(1)},\ldots,\lambda_{(k)}$ to be common first integrals of all the vector fields $\gamma_{(1)},\ldots,\gamma_{(k)}$. Hence, 
$$
\mathcal{L}_{\gamma_{(s)}}\lambda^{(s')}_i=0,\qquad 1\leq s\neq s'\leq k,\qquad i=1,\ldots,w,
$$
and the final integrability condition for $\lambda_{(1)},\ldots,\lambda_{(k)}$ is satisfied. Therefore, 
one obtains $k$-wave solutions. 

Let us give an example of the previous procedure based on the three-dimensional Lie algebra $\mathfrak{r}_{3,-1}$ given in Table \ref{fig:Lie-algebra-classification}. There exists a matrix representation of the Lie algebra $\mathfrak{r}_{3,-1}$ of the form
$$
M_1=
\frac 12\left[\begin{array}{cccc}
-1&0&0&0\\
0&1&0&0\\
0&0&-1&0\\
0&0&0&1\\
\end{array}\right],\quad M_2=\left[\begin{array}{cccc}
0&0&0&0\\
-1&0&0&0\\
0&0&0&0\\
0&0&0&0\\
\end{array}\right],\quad M_3=\left[\begin{array}{cccc}
0&0&0&0\\
0&0&0&0\\
0&0&0&-1\\
0&0&0&0\\
\end{array}\right].
$$
Indeed,
$$
[M_1,M_2]=M_2,\qquad [M_1,M_3]=-M_3,\qquad [M_2,M_3]=0
$$
has the same structure constants as the basis $\{e_1,e_2,e_3\}$ of $\mathfrak{r}_{3,-1}$  in Table \ref{fig:Lie-algebra-classification} that led to our model.
The associated vector fields on $\mathbb{R}^4$ are given by
\begin{equation}\label{eq:Gamma3}
\gamma_{(1)}=\frac 12\left(u^1\frac{\partial}{\partial u^1}-u^2\frac{\partial}{\partial u^2}+u^3\frac{\partial}{\partial u^3}-u^4\frac{\partial}{\partial u^4}\right),\qquad \gamma_{(2)}=u^1\frac{\partial}{\partial u^2},\qquad \gamma_{(3)}=u^4\frac{\partial}{\partial u^3}.
\end{equation}
Then,
\begin{equation}\label{eq:GammaRelCom}
[\gamma_{(1)},\gamma_{(2)}]=\gamma_{(2)},\qquad [\gamma_{(1)},\gamma_{(3)}]=-\gamma_{(3)},\qquad [\gamma_{(2)},\gamma_{(3)}]=0
\end{equation}
are the commutation relations for $\mathfrak{r}_{3,-1}$ in the chosen basis representing the matrix elements $M_1,M_2,M_3$which satisfy the same commutation relations (\ref{eq:GammaRelCom}). The distribution spanned by $\gamma_{(1)},\gamma_{(2)},\gamma_{(3)}$ has rank three almost everywhere, namely $\gamma_{(1)}\wedge\gamma_{(2)}\wedge\gamma_{(3)}\neq 0$ almost everywhere, and its annihilator is spanned, almost everywhere, by
$$
u^4du^1+u^1du^4=d(u^1u^4).
$$
Then, any function $f=f(u^1u^4)$ is a first integral of $\gamma_{(1)},\gamma_{(2)},\gamma_{(3)}$. Moreover, one can choose $\lambda^{(1)},\lambda^{(2)},\lambda^{(3)}$ as differential one-forms with coefficients given by first integrals of $\gamma_{(1)},\gamma_{(2)},\gamma_{(3)}$. It is simple to obtain a system to construct a three-wave. For instance, consider
$$
\lambda^{(1)}=u^1u^4dx^1,\qquad \lambda^{(2)}=u^1u^4dx^2,\qquad \lambda^{(3)}=u^1u^4dx^3.
$$
as differential one-forms on the space of independent variables $\mathbb{R}^4$ with coefficients in the space of dependent variables $\mathbb{R}^4$.
Then, $\lambda^{(1)}\wedge \lambda^{(2)}\wedge \lambda^{(3)}\neq 0$ and
$$
\mathcal{L}_{\gamma_{(s')}}\lambda^{(s)}=0,\qquad 1\leq s\neq s'\leq 3.
$$
Both previous conditions can easily be achieved by enlarging the dimension of the space of independent variables, which can be done with no restrictions, and due to the fact that the coefficients of the $\lambda^{(1)},\lambda^{(2)},\lambda^{(3)}$, are first integrals of $\gamma_{(1)},\gamma_{(2)},\gamma_{(3)}$.  Moreover, the system of PDEs we are analysing is given by
$$
\sum_{i=1}^4[A^i(u^1u^4)u^4,0,0,A^i(u^1u^4)u^1]\frac{\partial}{\partial x^i}\left[\begin{array}{c}u^1\\u^2\\u^3\\u^4\end{array}\right]=0.
$$

Let us now study a system of the form (\ref{eq:ExSubm}), where $\gamma_{(1)},\gamma_{(2)},\gamma_{(3)}$ are given in (\ref{eq:Gamma3}). In other words, we are interested in determining an integrable system of PDEs of the form
\begin{equation}\label{eq:LabelEx}
\frac{\partial u}{\partial r_\nu}=\sum_{\mu=1}^3f_{\nu\mu}\gamma_{(\mu)}(u),\qquad u=(u^1,u^2,u^3,u^4)\in \mathbb{R}^4,\qquad \nu=1,\ldots,3.\end{equation}

To be integrable, one has to obey the conditions (\ref{eq:IntCon}) for our system of PDEs. 
In particular,
one obtains
\begin{equation}\label{eq::ParSys}
    \begin{gathered}
\frac{\partial f_{\mu1}}{\partial r_\alpha}-\frac{\partial f_{\alpha 1}}{\partial r_\mu}=0,\qquad \frac{\partial f_{\mu2}}{\partial r_\alpha}-\frac{\partial f_{\alpha 2}}{\partial r_\mu}+f_{\mu2}f_{\alpha1}-f_{\alpha2}f_{\mu1}=0,\\
\frac{\partial f_{\mu3}}{\partial r_\alpha}-\frac{\partial f_{\alpha 3}}{\partial r_\mu}-f_{\mu3}f_{\alpha1}+f_{\alpha3}f_{\mu1}=0,
\end{gathered}
\end{equation}
for $1\leq \mu<\alpha\leq 3$. Then, there exists a function $g\in C^\infty(\mathbb{R}^3)$ such that
$$
f_{\alpha1}=\frac{\partial g}{\partial r_\alpha},\qquad \alpha=1,2,3.
$$
Thus, let us consider a solution for the remaining equations in (\ref{eq::ParSys}) assuming
$$
0=\frac{\partial f_{\mu2}}{\partial r_\alpha}+f_{\mu2}f_{\alpha 1}\Rightarrow 0=\frac{\partial}{\partial r_\alpha}(\ln f_{\mu2}+g)\Rightarrow f_{\mu2}=e^{-g}\kappa_{\mu2}, \qquad \mu=1,2,3,
$$
for certain constants $\kappa_{12},\kappa_{22},\kappa_{32}$. Similarly, 
$$
0=\frac{\partial f_{\mu3}}{\partial r_\alpha}-f_{\mu3}f_{\alpha 1}\Rightarrow 0= \frac{\partial}{\partial r_\alpha}(\ln f_{\mu3}-g)=0\Rightarrow f_{\mu3}=e^{g}\kappa_{\mu3},\qquad \mu=1,2,3,
$$
for some constants $\kappa_{13},\kappa_{23},\kappa_{33}$. Hence,
one can consider the matrix of coefficients of (\ref{eq:LabelEx}) are
$$
\left[\begin{array}{ccc}\frac{\partial g}{\partial r_1}&e^{-g}\kappa_{12}&e^{g}\kappa_{13}\\
\frac{\partial g}{\partial r_2}&e^{-g}\kappa_{22}&e^{g}\kappa_{23}\\
\frac{\partial g}{\partial r_3}&e^{-g}\kappa_{32}&e^{g}\kappa_{33}\\\end{array}\right].
$$
In particular, assume $g=r_1$. The above coefficient matrix becomes
$$
\left[\begin{array}{ccc}1&0&0\\
0&e^{-r_1}&0\\
0&0&e^{r_1}\\\end{array}\right]
$$
and the associated system of PDEs for the three-wave under study is
\begin{equation}\label{eq:SystemRec}
\frac{\partial u}{\partial r_1}=\gamma_{(1)},\qquad \frac{\partial u}{\partial r_2}=e^{-r_1}\gamma_{(2)},\qquad \frac{\partial u}{\partial r_3}=e^{r_1}\gamma_{(3)}.
\end{equation}
It is worth noting that there is another approach to the  construction of such a system of PDEs which involves putting the vector fields $\gamma_{(1)},\gamma_{(2)},\gamma_{(3)}$ into a quasi-rectifiable form, which was discussed previously. Note that, in view of the commutation relations (\ref{eq:GammaRelCom}), one has that
$$
\left[\frac{\partial}{\partial r_1}+\gamma_{(1)},\frac{\partial}{\partial r_2}+e^{-r_1}\gamma_{(2)}\right]\!=\!\left[\frac{\partial}{\partial r_1}+\gamma_{(1)},\frac{\partial}{\partial r_3}+e^{r_1}\gamma_{(3)}\right]\!=0
$$
and
$$
\!\left[\frac{\partial}{\partial r_2}+e^{-r_1}\gamma_{(2)},\frac{\partial}{\partial r_3}+e^{r_1}\gamma_{(3)}\right]=0
$$
and the system (\ref{eq:SystemRec}) is integrable. It is also worth noting that it is simple to obtain some coefficients depending on $r_1,r_2,r_3$ to multiply $\gamma_{(1)},\gamma_{(2)},\gamma_{(3)}$ and make them commute. 

This is different from the standard method, where we multiply $\gamma_{(1)},\gamma_{(2)},\gamma_{(3)}$ by functions on the space of dependent variables. The previous method can be applied to all quasi-rectifiable Lie algebras of vector fields detailed in the classification of Section \ref{Sec::ClaRecLieAlg}.

\section*{Acknowledgements}

A.M. Grundland was partially supported by an Operating Grant from NSERC of Canada.  J. de Lucas acknowledges a Simons--CRM professorship funded by the Simons Foundation and the Centre de Recherches Math\'ematiques (CRM) of the Universit\'e de Montr\'eal. J. de Lucas also acknowledges partial financial support provided by the Universit\'e du Qu\'ebec \`a Trois-Rivi\`eres for his visit to the CRM. 

\newpage
\section{Appendix: Classification of quasi-rectifiable indecomposable Lie algebras}

Let us summarise our classification of four- and five-dimensional quasi-rectifiable indecomposable Lie algebras. Our results follow from a simple but long application of the analysis of equation \eqref{eq:Charact}. Additional details on the parameters of indecomposable Lie algebras in the following tables can be found in \cite{WS14}. In any case, the specific values of such parameters are not relevant to the applications and results analysed in this work.
{\small 
%\begin{landscape}
\begin{table}[ht]
   \centering
    \begin{tabular}{|c|c|c|c|c|c|c|c|c|c|c|}
        \hline
         $\mathfrak{g}$& $[e_1,e_2]$ & $[e_1,e_3]$ & $[e_1,e_4]$ & $[e_2,e_3]$&$[e_2,e_4]$&$[e_3,e_4]$& Polyn. & quasi.\!\!\! rect. \\
 \hline
        $\mathfrak{h}_2\oplus \mathbb{R}^2$ & $e_2$ & $0$ & $0$ & $0$ &$0$&$0$& No & Yes\\
 \hline
        $\mathfrak{h}_2\oplus \mathfrak{h}_2$ & $e_2$ & $0$ & $0$ & $0$ &$0$&$e_4$& No & Yes\\
        \hline
        $\mathfrak{h}_3\oplus \mathbb{R}$ & $e_3$ & $0$ & $0$ & $0$ &$0$&$0$&$P_{123}$& No\\\hline
        $\mathfrak{n}_{4,1}$ & $0$ & $0$ & $0$ & $0$ &$e_1$&$e_2$& $P_{124}$& No\\
        \hline
             $\mathfrak{s}_{4,1} $ & $0$ & $0$ & $0$ & $0$ &$-e_1$&$-e_3$&$P_{124}$& No\\
              \hline
             $\mathfrak{s}_{4,2} $ & $0$ & $0$ & $-e_1$ & $0$ &$-e_1-e_2$&$-e_2-e_3$& $P_{123}$&No\\
        \hline
             $\mathfrak{s}_{4,3} $ & $0$ & $0$ & $-e_1$ & $0$&$-ae_2$ &$-be_3$&No& Yes\\
             \hline
             $\mathfrak{s}_{4,4} $ & $0$ & $0$ & $-e_1$ & $0$ &$-e_1-e_2$&$-ae_3$&$P_{124}$& No\\
        \hline
             $\mathfrak{s}_{4,5} $ & $0$ & $0$ & $-\alpha e_1$ & $0$ &$e_3-\beta e_2$&$-e_2-\beta e_3$&$P_{234}$& No\\
                     \hline
             $\mathfrak{s}_{4,6} $ & $0$ & $0$ & $0$ & $e_1$ &$-e_2$&$e_3$&$P_{123}$& No\\
                     \hline
             $\mathfrak{s}_{4,7} $ & $0$ & $0$ & $0$ & $e_1$ &$e_3$&$-e_2$&$P_{123}$& No\\
                     \hline
             $\mathfrak{s}_{4,8} $ & $0$ & $0$ & $-(1+a)e_1$ & $e_1$ &$ -e_2$&$-ae_3$&$P_{123}$& No\\        \hline
             $\mathfrak{s}_{4,9} $ & $0$ & $0$ & $-2\alpha e_1$ & $e_1$ &$e_3-\alpha e_2$&$-e_2-\alpha e_3$& $P_{123}$&No\\
                     \hline
             $\mathfrak{s}_{4,10} $ & $0$ & $0$ & $-2e_1$ & $e_1$ &$-e_2$&$-e_2-e_3$&$P_{123}$& No\\
                     \hline
             $\mathfrak{s}_{4,11} $ & $0$ & $0$ & $-e_1$ & $e_1$ &$-e_2$&$0$& $P_{123}$&No\\
                \hline
             $\mathfrak{s}_{4,12} $ & $0$ & $-e_1$ & $e_2$ & $-e_2$ &$-e_1$&$0$& $P_{124}$&No\\
        %$\mathfrak{h}_3\oplus \mathbb{R}$ & $e_3$ & $0$ & $0$ & $0$ &$0$&$0$&$0$& No\\
        \hline
    \end{tabular}
    \caption{Classification of quasi-rectifiable non-Abelian indecomposable four-dimensional Lie algebras. Note that $\lambda\in (-1,1).$ The value of the relevant polynomial coefficients of $\vartheta\wedge \delta \vartheta$ for $\vartheta=\sum_{i=1}^4\lambda_ie^i$ for the basis $\{e^1,e^2,e^3,e^4\}$ dual to the basis $\{e_1,e_2,e_3,e_4\}$ of the Lie algebra $\mathfrak{g}$ is given. Note that $\mathfrak{h}_3$ is the Heisenberg Lie algebra. Only one of the polynomial coefficients of $\vartheta\wedge \delta\vartheta$ is necessary in order to show that there is no quasi-rectifiable basis. The polynomial $P_{123}$ is $2\lambda_2^2$ for every Lie algebra that is not quasi-rectifiable, except for $\mathfrak{h}_3\oplus \mathbb{R}$, which has $P_{123}=-2\lambda_3^2$. The polynomial $P_{124}$ is proportional to $\lambda_1^2$ for $\mathfrak{n}_{4,1}$,  $\mathfrak{s}_{4,1}$, and $\mathfrak{s}_{4,4}$, while it is proportional to $\lambda_1^2+\lambda_2^2$ for $\mathfrak{s}_{4,12}$. Finally, $P_{234}$ is proportional to $\lambda_2^2+\lambda_3^2$. The coefficients $a,b,\alpha$ take different values, which are of not relevant in this work (see \cite{WS14} for details).}
    \label{fig:Lie-algebra-classification4}
\end{table}}
%\end{landscape}
\begin{landscape}

\begin{table}[ht]
   % \centering
   % \captionsetup{font=tiny} % Cambia el tamaño de fuente para el texto de la tabla
   {\small \begin{tabular}{|c|c|c|c|c|c|c|c|c|c|}
         \hline
      $\mathfrak{g}$& $[e_1,e_5]$ & $[e_2,e_3]$&$[e_2,e_4]$&$[e_2,e_5]$&$[e_3,e_4]$&$[e_3,e_5]$& $[e_4,e_5]$&Rec Pol. & quasi-rectifiable\\
 \hline
        $\mathfrak{n}_{5,1}$ & $0$ &$0$&$0$&$0$&$0$ &$e_1$&$e_2$&$P_{235}$& No\\
 \hline
        $\mathfrak{n}_{5,2}$  & $0$ &$0$&$0$&$0$&$e_2$ &$e_1$&$e_3$&$P_{235}$& No\\
 \hline
        $\mathfrak{n}_{5,3}$  & $0$ &$0$&$e_1$&$0$&$0$ &$e_1$&$0$&$P_{124}$& No\\
 \hline
        $\mathfrak{n}_{5,4}$ & $0$ &$0$&$0$&$e_1$&$e_1$ &$0$&$e_2$&$P_{125}$& No\\
   \hline
        $\mathfrak{n}_{5,5}$ & $0$ &$0$&$0$&$e_1$&$0$ &$e_2$&$e_3$&$P_{125}$& No\\
   \hline
        $\mathfrak{n}_{5,6}$  & $0$ &$0$&$0$&$e_1$&$e_1$ &$e_2$&$e_3$&$P_{125}$& No\\
        \hline
             $\mathfrak{s}_{5,1} $  & $0$ &$0$&$0$&$-e_1$&$0$&$-e_2$&$-e_4$&$P_{125}$& No\\
              \hline
             $\mathfrak{s}_{5,2} $  & $0$ &$0$&$0$&$-e_1$&$0$&$-e_3$&$-e_3-e_4$&$P_{125}$& No\\
          \hline
             $\mathfrak{s}_{5,3} $  & $0$ &$0$&$0$&$-e_1$&$0$&$-e_3$&$-ae_4$&$P_{125}$& No\\
          \hline
             $\mathfrak{s}_{5,4} $  & $0$ &$0$&$0$&$-e_1$&$0$&$e_4-\alpha e_3$&$-e_3-\alpha e_4$&$P_{125}$& No\\
          \hline
             $\mathfrak{s}_{5,5} $  & $-e_1$&$0$&0&$-e_1-e_2$&$0$&$-e_3-e_2$&$-e_3-e_4$&$P_{125}$& No\\
          \hline
             $\mathfrak{s}_{5,6} $  & $-e_1$ &$0$&$0$&$-e_1-e_2$&$0$&$-ae_3$&$-e_3-ae_4$&$P_{125}$& No\\
          \hline
             $\mathfrak{s}_{5,7} $ & $-e_1$ &$0$&0&$-e_1-e_2$&$0$&$-e_2-e_3$&$-ae_4$&$P_{125}$& No\\
             \hline
             $\mathfrak{s}_{5,8} $  & $e_2-\alpha e_1$ &$0$&$0$&$-\alpha e_2-e_1$&$0$&$-e_1-e_4-\alpha e_3$&$-e_2-e_3-\alpha e_4$&$P_{125}$& No\\        \hline
             $\mathfrak{s}_{5,9} $  & $-e_1$ &$0$&$0$&$-a e_2$&$0$&$-be_3$&$-ce_4$&No& Yes\\
                     \hline
             $\mathfrak{s}_{5,10} $ &$-e_1$ & 0 &0&$-e_1-e_2$&$0$&$-a e_3$&$-be_4$&$P_{125}$& No\\
                     \hline
             $\mathfrak{s}_{5,11} $  &$-\alpha e_1$& 0 &0&$-\beta e_2$&$0$&$e_4-\gamma e_3$&$-e_3-\gamma e_4$&No& Yes\\
                \hline
             $\mathfrak{s}_{5,12} $  & $-e_1$ &$0$&$0$&$-e_1-e_2$&$0$&$\beta e_4-\alpha e_3$&$-\beta e_3-\alpha e_4$&$P_{125}$& No\\\hline
             $\mathfrak{s}_{5,13} $  & $e_2-\alpha e_1$ &$0$&$0$&$-e_1-\alpha e_2$&$0$&$\gamma e_4-\beta e_3$&$-\gamma e_3-\beta e_4$&$P_{125}$& No\\\hline
             $\mathfrak{s}_{5,14} $ & $0$ &$e_1$&$0$&$0$&$0$&$-e_2$&$-e_4$&$P_{123}$& No\\\hline
             $\mathfrak{s}_{5,15} $  & $0$ &$e_1$&$0$&$-e_2$&$0$&$e_3$&$-e_1$&$P_{123}$& No\\\hline
             $\mathfrak{s}_{5,16} $  & $0$ &$e_1$&$0$&$e_3$&$0$&$-e_2$&$-e_1$&$P_{123}$& No\\\hline
             $\mathfrak{s}_{5,17} $  & $0$ &$e_1$&$0$&$-e_2$&$0$&$e_3$&$-ae_4$&$P_{123}$& No\\\hline
             $\mathfrak{s}_{5,18} $  & $0$ &$e_1$&$0$&$e_2$&$0$&$-e_3-e_4$&$-e_4$&$P_{123}$& No\\\hline
             $\mathfrak{s}_{5,19} $  & $0$ &$e_1$&$0$&$e_3$&$0$&$-e_2$&$-\alpha e_4$&$P_{123}$& No\\\hline
             $\mathfrak{s}_{5,20} $  & $-2e_1$&$e_1$&$0$&$-e_2$&$0$&$-e_4$&$0$&$P_{123}$& No\\
        %$\mathfrak{h}_3\oplus \mathbb{R}$ & $e_3$ & $0$ & $0$ & $0$ &$0$&$0$&$0$& No\\
        \hline
    \end{tabular}
    \caption{First part of the classification of quasi-rectifiable non-Abelian indecomposable five-dimensional Lie algebras. In this first table, $[e_1,e_\alpha]=0$ for $\alpha=2,3,4$. Only the coefficient of equation (\ref{eq:Charact}) is needed to prove that there is no quasi-rectifiable basis. The polynomial coefficient is proportional to $\lambda_1^2$ except for the cases $\mathfrak{s}_{5,8}$ and $\mathfrak{s}_{5,13}$, where the polynomial is proportional to $\lambda_1^2+\lambda_2^2$.}
    \label{fig:Lie-algebra-classification5}}
\end{table}
\end{landscape}

\begin{landscape}
{\tiny 
\begin{table}[ht]
    \centering
   % \captionsetup{font=tiny} % Cambia el tamaño de fuente para el texto de la tabla
  {\small  \begin{tabular}{|c|c|c|c|c|c|c|c|c|c|c|c|}
         \hline
      $\mathfrak{g}$&  $[e_1,e_4]$ & $[e_1,e_5]$ & $[e_2,e_3]$&$[e_2,e_4]$&$[e_2,e_5]$&$[e_3,e_4]$&$[e_3,e_5]$& $[e_4,e_5]$&Polyn. & quasi-rect. \\\hline
             $\mathfrak{s}_{5,21} $ &  $0$ & $-2e_1$ &$e_1$&$0$&$-e_2-e_3$&$0$&$-e_3-e_4$&$-e_4$&$P_{123}$& No\\\hline
             $\mathfrak{s}_{5,22} $ &  $0$ & $-(a+1)e_1$ &$e_1$&$0$&$-e_2$&$0$&$-ae_3$&$-b e_4$&$P_{123}$& No\\\hline
             $\mathfrak{s}_{5,23} $ &  $0$ & $-(a+1)e_1$ &$e_1$&$0$&$-ae_2$&$0$&$-e_4-e_3$&$-e_4$&$P_{123}$& No\\\hline
             $\mathfrak{s}_{5,24} $ &  $0$ & $-2e_1$ &$e_1$&$0$&$-e_2-e_3$&$0$&$-e_3$&$-a e_4$&$P_{123}$& No\\\hline
             $\mathfrak{s}_{5,25} $ &  $0$ & $-2\alpha e_1$ &$e_1$&$0$&$e_3-\alpha e_2$&$0$&$-e_2-\alpha e_3$&$-\beta e_4$&$P_{123}$& No\\\hline
             $\mathfrak{s}_{5,26} $ &  $0$ & $-(a+1)e_1$ &$e_1$&$0$&$-e_2$&$0$&$-ae_3$&$-e_1-(a+1)e_4$&$P_{123}$& No\\\hline
             $\mathfrak{s}_{5,27} $ & $0$ & $-2e_1$ &$e_1$&$0$&$-e_2-e_3$&$0$&$-e_3$&$-e_1-2e_4$&$P_{123}$& No\\\hline
             $\mathfrak{s}_{5,28} $ & $0$ & $-2\alpha e_1$ &$e_1$&$0$&$-\alpha e_2-e_3$&$0$&$e_2-\alpha e_3$&$-e_1-2\alpha e_4$&$P_{123}$& No\\\hline
             $\mathfrak{s}_{5,29} $ &  $0$ & $-e_1$ &$e_1$&$0$&$0$&$0$&$-e_4-e_3$&$-e_4$&$P_{123}$& No\\\hline     
             $\mathfrak{s}_{5,30} $ & $0$ & $-e_1$ &$e_1$&$0$&$-e_2$&$0$&$0$&$-ae_4$&$P_{123}$& No\\\hline
             $\mathfrak{s}_{5,31} $ & $0$ & $-e_1$ &$e_1$&$0$&$-e_2$&$0$&$0$&$-e_1-e_4$&$P_{123}$& No\\\hline
             $\mathfrak{s}_{5,32} $ &$0$ &$-e_1$ &$e_1$&$0$&$0$&$-e_3-e_4$&$-e_4-e_1$&$0$&$P_{123}$& No\\\hline
             $\mathfrak{s}_{5,33} $ &  $0$ & $0$ &$0$&$e_1$&$e_2$&$e_2$&$2e_3$&$-e_4$&$P_{123}$& No\\\hline
             $\mathfrak{s}_{5,34} $ &  $0$ & $-3e_1$ &$0$&$e_1$&$-2e_2$&$e_2$&$-e_3$&$-e_4-e_3$&$P_{124}$& No\\\hline
             $\mathfrak{s}_{5,35} $ & $0$ & $-(a+2)e_1$ &$0$&$e_1$&$-(a+1)e_2$&$e_2$&$-ae_3$&$-e_4$&$P_{124}$& No\\\hline
             $\mathfrak{s}_{5,36} $ &  $0$ & $-2e_1$ &$0$&$e_1$&$-e_2$&$e_2$&$0$&$-e_4$&$P_{124}$& No\\\hline
             $\mathfrak{s}_{5,37} $ &  $0$ & $-e_1$ &$0$&$e_1$&$-e_2$&$e_2$&$-e_3$&$0$&$P_{124}$& No\\\hline
             $\mathfrak{s}_{5,38} $ & $0$ & $-e_1$ &$0$&$e_1$&$-e_2$&$e_2$&$-\epsilon e_1-e_3$&$0$&$P_{124}$& No\\\hline
               $\mathfrak{s}_{5,39} $ & $0$ & $0$ &$0$&$-e_2$&$0$&$0$&$-e_3$&$-e_1$&$P_{145}$& No\\\hline
             $\mathfrak{s}_{5,40} $ &  $0$ & $0$ &$0$&$-e_2$&$e_3$&$-e_3$&$-e_2$&$-e_1$&$P_{145}$& No\\\hline
             $\mathfrak{s}_{5,41} $ &$-e_1$ & $0$ &$e_1$&$0$&$-e_2$&$-ae_3$&$-be_3$&$0$&$P_{123}$& No\\\hline
             $\mathfrak{s}_{5,42} $ &  $-ae_1$ & $-e_1$ &0&$-e_2$&$0$&$-e_3$&$-e_2$&$0$&$P_{123}$& No\\\hline
             $\mathfrak{s}_{5,43} $ &  $-\alpha e_1$ & $-\beta e_1$ &$0$&$-e_2$&$e_3$&$-e_3$&$-e_2$&$0$&$P_{235}$& No\\\hline
             $\mathfrak{s}_{5,44} $ &  $-e_1$ & $0$ &$e_1$&$-e_2$&$-e_2$&$0$&$e_3$&$0$&$P_{123}$& No\\\hline
             $\mathfrak{s}_{5,45} $ &  $-2e_1$ & $0$ &$e_1$&$-e_2$&$-e_3$&$-e_3$&$e_2$&$0$&$P_{123}$& No\\\hline
             2$\mathfrak{n}_{1,1} $ &  $e_5$ & $0$ &$2e_3$&$e_4$&$-e_5$&$0$&$e_4$&$0$&$P_{145}$& No\\\hline
    \end{tabular}
    \caption{Continuation of the classification of quasi-rectifiable non-Abelian indecomposable five-dimensional Lie algebras. In this table, $[e_1,e_2]=[e_1,e_3]=0$ except for $2\mathfrak{n}_{1,1}$, where $[e_1,e_2]=2e_1$ and $[e_1,e_3]=-e_2$. Note that $\lambda\in (-1,1).$ Only one coefficient of equation (\ref{eq:Charact}) is needed to prove that there is no quasi-rectifiable basis. The written polynomial coefficient is proportional to $\lambda_1^2$ except for the cases $2\mathfrak{n}_{1,1}$ and $\mathfrak{s}_{5,43}$, whose coefficients are proportional to the polynomials $\lambda_2^2+\lambda_3^2$ and $\lambda_5^2$ respectively.}
    \label{fig:Lie-algebra-classification6}}
  %  \vspace{-Xpt}
\end{table}}
\end{landscape}

\end{document}